\documentclass[a4paper,UKenglish,cleveref,thm-restate]{lipics-v2021}

\hideLIPIcs  %

\usepackage{booktabs}
\usepackage[linesnumbered,ruled,vlined]{algorithm2e}
\RestyleAlgo{boxruled}
\NoCaptionOfAlgo
\usepackage{dsfont}

\title{Near-optimal population protocols on bounded-degree trees} %

\author{Joel Rybicki}{Humboldt-Universität zu Berlin,  Germany}{joel.rybicki@hu-berlin.de}{https://orcid.org/0000-0002-6432-6646}{}
\author{Jakob Solnerzik}{Humboldt-Universität zu Berlin, Germany}{jakob.solnerzik@hu-berlin.de}{https://orcid.org/0009-0008-2584-2171}{}
\author{Robin Vacus}{Sorbonne Université, CNRS, LIP6, France}{robin.vacus@cnrs.fr}{https://orcid.org/0000-0002-7368-4912}{}

\authorrunning{J. Rybicki, J. Solnerzik, and R. Vacus}

\Copyright{Joel Rybicki, Jakob Solnerzik, and Robin Vacus} %



\category{} %
\nolinenumbers

\funding{This research was funded by the Deutsche Forschungsgemeinschaft (DFG, German Research Foundation) – Project number 539576251.}

\newcommand{\fout}{\mu}

\newcommand{\col}{\textup{\texttt{colour}}}
\newcommand{\parent}{\textup{\texttt{parent}}}
\newcommand{\children}{\textup{\texttt{children}}}
\newcommand{\stamp}{\textup{\texttt{stamp}}}
\newcommand{\token}{\textup{\texttt{token}}}

\newcommand{\out}{\textup{\texttt{output}}}
\newcommand{\colouringalgo}{\hyperref[sec:2-hop_colouring_algo]{2-hop colouring algorithm}\xspace}
\newcommand{\mainalgo}{\hyperref[alg:tree-orientation]{orientation algorithm}\xspace}
\newcommand{\orientalgo}{\textup{\texttt{Set-Edge-Orientation}}\xspace}
\newcommand{\majalgo}{\hyperref[sec:maj_algorithm]{majority algorithm}\xspace}
\newcommand{\uniform}{\textrm{Uniform}}
\newcommand{\C}{\mathcal{C}}
\newcommand{\R}{\mathcal{R}}

\newcommand{\M}{\mathcal{M}}
\newcommand{\Z}{\mathcal{Z}}
\newcommand{\calL}{\mathcal{L}}

\newcommand{\depth}{\mathrm{depth}}

\newcommand{\diam}{\mathrm{diam}}

\newcommand{\atype}{\mathsf{A}}
\newcommand{\btype}{\mathsf{B}}
\newcommand{\ctype}{\mathsf{C}}

\newcommand{\bbOne}{\mathds{1}}
\newcommand{\pa}[1]{\left( #1 \right)}

\newcommand{\Exp}{\mathbb{E}}
\newcommand{\E}{{\rm I\kern-.3em E}}

\newcommand{\dis}{\mathrm{dist}}
\newcommand{\ap}{^\prime}

\DeclareMathOperator{\bbN}{\mathbb{N}}

\DeclareMathOperator{\dist}{dist}
\DeclareMathOperator*{\Geom}{Geom}
\DeclareMathOperator*{\poly}{poly}
\DeclareMathOperator*{\polylog}{polylog}

\newcommand{\stconflict}{C_t^{\mathrm{stamp}}}
\newcommand{\coconflict}{C_t^{\mathrm{colour}}}

\begin{document}

\maketitle

\begin{abstract}
  We investigate space-time trade-offs for population protocols in sparse interaction graphs. In complete interaction graphs, optimal space-time trade-offs are known for the leader election and exact majority problems. However, it has remained open whether other graph families exhibit similar space-time complexity trade-offs, as existing lower bound techniques do not extend beyond highly dense graphs.

  In this work, we show that -- unlike in complete graphs -- population protocols on bounded-degree trees do not exhibit significant asymptotic space-time trade-offs for leader election and exact majority. For these problems, we give constant-space protocols that have near-optimal worst-case expected stabilisation time. These new protocols achieve a linear speed-up compared to the state-of-the-art.

  Our results are based on two novel protocols, which we believe are of independent interest. First, we give a new fast self-stabilising 2-hop colouring protocol for general interaction graphs, whose stabilisation time we bound using a stochastic drift argument. Second, we give a self-stabilising tree orientation algorithm that builds a rooted tree in optimal time on any tree. As a consequence, we can use simple constant-state protocols designed for directed trees to solve leader election and exact majority fast.
  For example, we show that ``directed'' annihilation dynamics solve exact majority in $O(n^2 \log n)$ steps on directed trees.
\end{abstract}

\section{Introduction}

Population protocols are a model of asynchronous %
computation in systems with anonymous, identical, and computationally restricted agents~\cite{angluin2006computation,aspnes2009introduction}.
In this model, the agents reside on the nodes of a connected $n$-node graph $G=(V,E)$. The computation is controlled by a scheduler that, in each time step, selects an edge $\{u,v\} \in E$ uniformly at random; the two nodes $u$ and $v$ then interact by reading each other's current states and updating their local states.

The \emph{time complexity} is the expected number of global interactions required to reach a configuration with correct and stable outputs that cannot change no matter which interactions occur in the future. The \emph{space complexity} of a protocol is measured by the number of distinct states an agent can take.

\subparagraph{Space-time complexity thresholds on cliques.}
The past decade has seen a focused research program investigating space-time complexity trade-offs for population protocols on \emph{complete} graphs~\cite{doty_stable_2018,alistarh2017time,alistarh2018space,sudo2019logarithmic,berenbrink_optimal_2020,berenbrink2021time,doty2022time}.
This research has identified problems that
exhibit asymptotic \emph{complexity thresholds}: protocols with space complexity below the threshold are a factor $\tilde{\Theta}(n)$ slower than time-optimal protocols, and conversely, optimal stabilisation time can be achieved by protocols that use a number of states above this threshold.
There are two prototypical problems with such complexity~thresholds:

\begin{itemize}
  \item For leader election, the threshold is at $\Theta(\log \log n)$ states: the problem can be solved in
    optimal $\Theta(n \log n)$ time with $\Theta(\log \log n)$ states~\cite{berenbrink_optimal_2020,sudo_leader_2020,gkasieniec2020enhanced}. Any $o(\log \log n)$-state protocol requires $\tilde{\Omega}(n^2)$
    time~\cite{alistarh2017time} and a
    constant-state protocol solves the problem in $\Theta(n^2)$ time~\cite{doty_stable_2018}.

  \item For exact majority, the threshold is widely conjectured to be at $\Theta(\log n)$ states.
    Alistarh et al.~\cite{alistarh2018space} showed that time-optimal protocols satisfying certain technical ``output-dominance'' and ``monotonicity'' properties -- which are satisfied by all known clique protocols -- need $\Omega(\log n)$ states. This is matched by the time-optimal protocol of Doty et al.~\cite{doty2022time} that uses $\Theta(\log n)$ states and $\Theta(n \log n)$ time.
    A constant-state protocol solves the problem in $\Theta(n^2)$ time~\cite{draief2012convergence,mertzios2017determining}.
\end{itemize}

\subparagraph{Graphical population protocols.}
Recent work has started to extend the study of space-time complexity trade-offs from cliques to general interaction graphs~\cite{alistarh2021fast,yokota2021time,yokota2023near,alistarh2025near,rybicki2026space}. On general graphs, time and space complexity typically depend on the expansion properties of the interaction graph.
So far, all state-of-the-art upper bounds for general graphs parallel those in the clique. Leader election~\cite{alistarh2025near,yokota2023near} and exact majority~\cite{rybicki2026space} protocols that use a super-constant number (typically polylogarithmic in $n$) of states  are a factor $\tilde{\Theta}(n)$ faster than known constant-state protocols.

However, it is not known whether a super-constant number of states is \emph{necessary} to attain such near-linear speed-up in sparse and/or high-diameter graphs.
The main obstacle is that the surgery technique, which is a key ingredient for proving space-time trade-offs in complete graphs~\cite{doty_stable_2018,alistarh2017time,alistarh2018space},
does not easily generalise to other graph families -- with the exception of dense random graphs~\cite{alistarh2025near}. %

Indeed, Alistarh et al.~\cite{alistarh2025near} pointed out that in \emph{star graphs} a leader can be elected in constant time and space, because the underlying graph helps break symmetry fast. This simple example rules out that space-time trade-offs must hold in \emph{all} graphs.
Nevertheless, they showed that the surgery technique sometimes remains applicable outside of cliques by proving
that constant-state leader election requires $\tilde{\Theta}(n^2)$ expected steps in dense random graphs (with constant expected diameter).

\subsection{Our focus: space-time trade-offs in sparse graphs}
In this work, we move away from highly \emph{dense} graphs to study the other extreme of highly \emph{sparse} interaction graphs, namely trees.
Trees turn out to be a particularly interesting class of graphs for population protocols:
\begin{enumerate}[(i)]

  \item
    Trees have much more heterogeneous structure than complete graphs:
    one cannot simply keep track of the \emph{state counts} in a configuration nor assume that all pairs of nodes can interact; this significantly complicates the design and analysis of  protocols.
  \item Trees are bad expanders: because trees have small cuts,  expansion-based population protocols and arguments for general graphs~\cite{alistarh2025near,rybicki2026space,alistarh2021fast}, which typically rely on fast mixing of tokens, fail to yield good time or space bounds in trees.

  \item  Trees have received significant recent attention in the standard LOCAL and CONGEST models of distributed graph algorithms~\cite{chang2020complexity,balliu2021improved,balliu2021locally,balliu2021congest,grunau2022landscape,balliu2023average,balliu2026distributed,baumecker2025nearly,brandt2025towards,balliu2024completing}, but no work has yet dealt specifically with trees in the more challenging population protocol~model.
\end{enumerate}
For bounded-degree trees, the state-of-the-art population protocols exhibit a $\tilde{\Omega}(n)$ factor difference between the state-of-the-art stabilisation times of constant-state and super-constant-state protocols for both leader election~\cite{beauquier2013self,alistarh2025near} and exact majority~\cite{draief2012convergence,rybicki2026space}.

For example, on paths, the known constant-state protocols for leader election and exact majority consensus use $\tilde{\Theta}(n^3)$ expected time, while protocols using a super-constant number of states can stabilise much faster in near-optimal $\tilde{O}{(n^2)}$ time~\cite{alistarh2025near,rybicki2026space}.
The information-propagation-based lower-bound techniques only give unconditional time lower bounds for arbitrary (i.e., including unbounded state) protocols~\cite{alistarh2025near,rybicki2026space}, so they cannot be used to rule out the existence of space-time trade-offs for trees or other graph classes.%

So far, it is not clear if we are simply missing the right lower bound techniques for proving space-time trade-offs in sparse graphs or if such  graphs exhibit fundamentally different behaviour.
In the light of the above, we investigate the following open question raised in prior work~\cite{alistarh2025near,rybicki2026space}:
\begin{quote}
  Do protocols in \emph{sparse} interaction graphs exhibit significant space-time trade-offs?
\end{quote}

\subsection{Population protocols on graphs} \label{sec:model_definition}

Let $G=(V,E)$ be a connected, undirected graph with $n$ nodes and $m$ edges.
The distance between two nodes $u$ and $v$ is $\dist(u,v)$.
We write $D$ for the diameter and $\Delta$ for the maximum degree of $G$.

\subparagraph{Population protocols.}
Let $\mathcal{G}$ be a class of graphs.
Formally, a \emph{probabilistic population protocol} on $\mathcal{G}$ consists of
a finite set $\Lambda$  of states, a state transition function $\Xi \colon \Lambda \times [0,1] \times \Lambda \times [0,1] \to \Lambda \times \Lambda$, and a map  $\fout \colon \Lambda \to \Gamma$ that maps the local state of a node to an output value in the set $\Gamma$.
Here, the  set $\Lambda$ can depend on the graph class, but not on the specific input graph $G$. For example, on bounded-degree graphs, the protocol can depend on the bound on the maximum degree $\Delta$.

Let $G  \in \mathcal{G}$. A stochastic \emph{schedule} on  $G=(V,E)$ is a random sequence $(e_t)_{t \ge 1}$ of edges, where each edge $e_t$ is sampled independently and uniformly at random from the set $E$ of edges.
A \emph{configuration} is a map $x \colon V \to \Lambda$. An \emph{execution} of the protocol under a schedule $\sigma = (e_t)_{t \ge 1}$ is a sequence $(x_t)_{t \ge 0}$ of configurations, where $x_0$ is the initial configuration and $x_{t}$ is defined inductively as
\[
  (x_{t}(u), x_t(v)) =
  \begin{cases}
    \Xi\big(x_{t-1}(u), R_t(u), x_{t-1}(v), R_t(v)\big) & \text{if } e_t = \{u,v\}, \\
    (x_{t-1}(u), x_{t-1}(v)) & \textrm{otherwise},
  \end{cases}
\]
where $R_t \in [0,1]^n$ is sampled independently and uniformly at random.
When node $v$ interacts at time step $t$, the value $R_t(v)$ is the local randomness used by the node $v$ during this interaction.
While the above definition is general, our protocols only employ $O(\log \Delta)$ random bits per interaction.
Here we do not explicitly assume the initiator-responder model, as the randomness given to the nodes can be used to break symmetry.
We say that a configuration $y$ is \emph{reachable} from a configuration $x$ if
\(
  \Pr[ x_{t} = y \mid x_0 = x ] > 0
\)
for some $t \ge 0$.
A configuration $x$ is \emph{stable} if any $y$ reachable from $x$ satisfies $\fout \circ y = \fout \circ x$.

The \emph{space complexity} of a protocol is the number $|\Lambda|$ of states.
The \emph{stabilisation time} of a protocol from an initial state $x_0$ is the random variable $T = \min \{ t : x_t \text{ is stable}\}$.
A protocol is \emph{self-stabilising} if the stabilisation time $T$ is finite almost~surely for any initial~configuration.

\subparagraph{Tasks.}
Let $\Sigma$ be a set of input labels and $\Gamma$ a set of output labels.
An \emph{input instance} is a pair $(G,f)$, where $f \colon V \to \Sigma$ is an input labelling and $G \in \mathcal{G}$.
A task $\Pi$ maps each input instance $(G, f)$ to a set of feasible solutions $\Pi(G, f)$, where each solution is a map $g \colon V \to \Gamma$.

We say that a protocol solves $\Pi$ on the input instance $(G,f)$ if the stabilisation time $T$ is finite almost surely and the output satisfies $\fout \circ x_{T} \in \Pi(G,f)$.

For any non-self-stabilising protocol for task $\Pi$, we assume without loss of generality that $\Sigma \subseteq \Lambda$ and the input $f$ gives the initial configuration, i.e., $x_0 = f$.
For self-stabilising protocols, the initial configuration $x_0$ is arbitrary and we only consider tasks without inputs (i.e., nodes do not have fixed local inputs that cannot be affected by transient faults).

\subsection{Our contributions}\label{sec:contributions}

\begin{figure}
  \includegraphics[width=\textwidth]{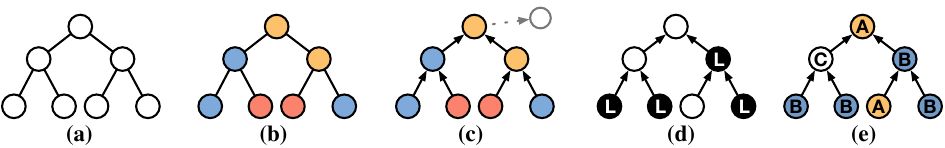}
  \caption{(a) An undirected input tree. (b) A 2-hop coloured tree. (c) The orientation protocol orients a 2-hop coloured tree. The grey node may be a ``hallucinated'' parent of the root. (d) An example configuration of the leader election protocol. The leader tokens ($\mathsf{L}$) are pushed towards the root. When two tokens meet, the child's token is removed.
  (e) An example configuration of the majority protocol. The $\mathsf{A}$- and $\mathsf{B}$-tokens are pushed towards the root; when $\mathsf{A}$- and $\mathsf{B}$-tokens meet, both turn into $\mathsf{C}$-tokens. In contrast to $\mathsf{A}$- and $\mathsf{B}$-tokens, the $\mathsf{C}$-tokens move away from the~root.}\label{fig:examples}
\end{figure}

We give near-time-optimal \emph{constant-state} protocols for exact majority and leader election in \emph{bounded-degree} trees. This gives the first example of a non-trivial class of graphs where (non-self-stabilising) population protocols do not exhibit significant asymptotic space-time trade-offs in terms of $n$.
On the way to obtain these results, we give new protocols we believe are of independent interest:
\begin{enumerate}%
  \item A fast self-stabilising 2-hop colouring protocol for general graphs.
  \item A time-optimal, self-stabilising protocol that orients 2-hop coloured trees.
\end{enumerate}
These two technical ingredients allow us to orient an undirected tree fast.
We then show that \emph{oriented} trees admit fast, constant-state protocols for leader election and exact majority.

We note that while we mostly focus on bounded-degree trees, i.e., the case $\Delta \in O(1)$, the exact majority and leader election protocols in fact work for \emph{any} undirected tree -- however, the space complexity is of order $2^{\poly(\Delta)}$, and similarly, the time complexities hide an additional $\Delta$-factor.
Unlike fast protocols for complete graphs~\cite{berenbrink_optimal_2020,doty2022time} and general graphs~\cite{alistarh2021fast,alistarh2025near,rybicki2026space}, our protocols do not need any information about the size $n$ of the population.

\subsubsection{Ingredient 1: Fast self-stabilising 2-hop colouring}
Our first ingredient is a new protocol for self-stabilising 2-hop colouring in \emph{general} graphs with maximum degree $\Delta$.
A 2-hop colouring of a graph $G=(V,E)$ is a function $f \colon V \to \{1, \ldots, \chi\}$, where $f(u) \neq f(v)$ for all  $u,v \in V$ such that~$u$ and~$v$ share a neighbour. That is, for any node $v \in V$, the neighbours of node~$v$ must have distinct colours.
A 2-hop colouring allows the nodes to distinguish their neighbours, and in a sense, it gives a ``port-numbering'' in the population protocol model; see~\Cref{fig:examples}b. Note that in a 2-hop colouring a node may have the same colour as (at most) one of~its~neighbours.

\begin{restatable}{theorem}{twohop}\label{thm:2-hop-colouring}
  There is a self-stabilising $2^{O(\Delta^2)}$-state protocol that stabilises on any graph to a 2-hop colouring with $O(\Delta^2)$ colours in $O( \Delta m \log n)$ time steps in expectation and with high probability.%
\end{restatable}

Our protocol is inspired by the 2-hop colouring algorithm of Angluin, Aspnes, Fischer, and Jiang~\cite{angluin_selfstabilising_2008}. They showed that their protocol always stabilises, but they did not analyse its time complexity under the stochastic scheduler.

In \Cref{sec:2-hop}, we design a new protocol that randomly samples colours when it detects a conflict. Like all existing population protocols for 2-hop colouring~\cite{angluin_selfstabilising_2008,Sudo2018_Loosely,Kanaya2024_Almost}, ours also uses ``stamps'' to detect if the condition for the 2-hop colouring is violated. %
However, our protocol comes with some additional tweaks -- namely, a simple stamp clearing mechanism -- to speed up its stabilisation.
With this, our protocol significantly improves on the time and space complexity of 2-hop colouring on sparse graphs; see \Cref{sec:related} for further discussion. %

For the time complexity analysis, we use a ``block-wise'' stochastic drift argument over sufficiently long time intervals. This gives us enough control on the drift to obtain a good bound on the stabilisation time. Our protocol is time-optimal up to the $\Delta$-factor: there is no protocol that stabilises in $o(m \log n)$ expected steps, as we show in \Cref{sec:2-hop_lower-bound}.
Here, the main challenge to carry out the drift analysis is to keep track of different types of conflicts so that they can be used to bound the drift; na\"ive analysis easily leads to unnecessary polynomial factors in $\Delta$.
In order to run this protocol, we assume nodes have access to $O(\log \Delta)$ random bits per interaction to sample a new colour from a palette of size $O(\Delta^2)$. Apart from this, all other protocols are deterministic.

\subsubsection{Ingredient 2: Time-optimal self-stabilising tree orientation}
We give a self-stabilising protocol that orients any 2-hop coloured diameter-$D$ tree in optimal $\Theta(\max\{D, \log n\} \cdot n)$ time with high probability. An \emph{orientation} of a tree assigns a direction for each edge such that all edges point towards a single root node $v \in V$.

In our protocol, each node maintains pointers to its parent and children using the 2-hop colouring. Prior to stabilisation, neighbouring nodes may have inconsistent state variables, and in such cases, the orientation of the corresponding edge is ill-defined. These inconsistencies are resolved upon each interaction by updating the pointers to correctly orient the edge. To do this, it is important to use a suitable local rule ensuring that ``badly'' oriented edges are rapidly pushed out of the tree. For this, the protocol requires some (arbitrary) symmetry-breaking information, such as an initiator-responder assignment for each interaction.

For the analysis, we use a potential argument that tracks the random walks performed by the badly oriented edges to show that such badly oriented edges are removed sufficiently fast.
This mechanism is formalised in \Cref{sec:orient}.

\begin{restatable}{theorem}{orient}\label{thm:coloured-tree-orientation}
  There is a self-stabilising protocol that orients
  any tree of diameter~$D$ that has been 2-hop coloured with at most $\chi$ colours
  in $O(n \cdot \max\{ D, \log n\})$ steps in expectation and with high probability
  using $O\left( 2^\chi \cdot \chi^2 \right)$ states.
\end{restatable}
Together with \Cref{thm:2-hop-colouring} that gives a 2-hop colouring with $\chi \in O(\Delta^2)$ colours, we get that a tree can be oriented using
$2^{O(\Delta^2)} \cdot 2^{O(\Delta^2)} \cdot \Delta^4 = 2^{O(\Delta^2)}$ states.
This gives us the next result:

\begin{restatable}{corollary}{orientsimp}\label{coro:orientation}
  Any undirected tree can be oriented in $O(Dn + \Delta n \log n)$ steps in expectation and with high probability using $2^{O(\Delta^2)}$ states.
\end{restatable}

\Cref{coro:orientation} implies that any bounded-degree tree can be oriented in time $O(Dn)$ with a constant-state protocol, because bounded-degree trees have $\Omega(\log n)$ diameter.
However, in the resulting orientation, the root does \emph{not} know that it is the root, i.e., it may think it has a parent; see \Cref{fig:examples}c for an example. Indeed, this is necessary: Angluin et al.~\cite{angluin_selfstabilising_2008} showed that
self-stabilising leader election is impossible in non-simple graphs, which include trees.
Thus, an orientation algorithm where the root knows it is the root would solve self-stabilising leader election on trees.
We show in \Cref{sec:lowerbounds} that our orientation protocol is essentially time-optimal by extending the lower-bound technique of Alistarh et al.~\cite{alistarh2025near}~to~trees.

\begin{restatable}{theorem}{orientlb}\label{thm:orientation-lb}
  There is no tree orientation protocol that stabilises in $o(Dn)$ expected steps. %
\end{restatable}

\subsubsection{Application 1: Optimal leader election for bounded-degree trees}

The tree orientation protocol allows us to easily construct a fast \emph{non-self-stabilising} stable leader election protocol.
In the stable leader election problem,
the task is to elect exactly one of the nodes as the (permanent) leader.
Our protocol is a ``directed'' variant of the classic 2-state leader election protocol on the clique~\cite{doty_stable_2018}.
Initially, all nodes start with a token. After the tree orientation protocol stabilises, a non-empty set of nodes still have a token. In the oriented tree, the tokens are always pushed towards the root, and when two tokens meet, only the token of the parent survives.
In \Cref{sec:tree-le}, we establish the following result.

\begin{restatable}{corollary}{leub}\label{coro:le}
  There is a leader election protocol for trees that uses
  $2^{O(\Delta^2)}$
  states and elects a leader in $O(Dn + \Delta n \log n)$ steps in expectation and with high probability.
\end{restatable}

\Cref{coro:le} implies that there exists a constant-state leader election protocol that stabilises in $O(Dn)$ time on
bounded-degree trees. In \Cref{sec:lowerbounds}, we show that this is optimal by establishing the following new lower bound for trees.

\begin{restatable}{theorem}{lelb}\label{thm:le-lb}
  There is no leader election protocol that stabilises in $o(Dn)$ expected steps in trees of diameter $D \in \Omega(\log n)$.
\end{restatable}

When put together, the above results show that on bounded-degree trees
leader election does not exhibit a space-time complexity threshold in terms of $n$.

\begin{corollary}\label{coro:bounded-le}
  On bounded-degree trees, there is a constant-state protocol that solves leader election in optimal $\Theta(Dn)$ time.
\end{corollary}

\subsubsection{Application 2: Fast space-efficient exact majority for trees }

Finally, we consider the exact majority task on oriented trees.
In this task, each node $v$ receives an input bit $f(v) \in \{0,1\}$, and the task is to reach consensus on the bit that was given as input to the majority of the nodes. In case of ties, consensus on an arbitrary bit should be reached.
On paths, the existing protocols for exact majority use $\tilde{\Omega}(n^3)$ time~\cite{draief2012convergence,rybicki2026space}.

We give a new protocol that stabilises in only $O(n^2 \log n)$ time on any tree, and on bounded-degree trees, uses only constant space.
Specifically, we consider a ``directed'' variant of the two-species annihilation dynamics recently studied by Rybicki et al.~\cite{rybicki2026space}.
Each node starts with a token $\atype$ or $\btype$, depending on its initial input bit.
When an $\atype$-token and a $\btype$-token meet, they both get removed, leaving behind two ``empty'' tokens of type $\ctype$. The tokens of type $\atype$ and $\btype$ are always pushed towards the root (so that the ``empty'' tokens move away from the root).
\Cref{fig:examples}e illustrates an example configuration of this~process.

We show that the directed annihilation dynamics stabilise in $O(n^2 \log n)$ steps into a configuration where only the majority input token remains (or in the case of a tie, only $\ctype$-tokens are left). In case of a majority, the root will hold one of the remaining input tokens.
The output is then propagated from the root to all nodes by having each node with a $\ctype$-token copy the output of its parent. This is explained in more detail in \Cref{sec:maj}.

\begin{restatable}{theorem}{majub} \label{thm:main_majority_upper_bound}
  There is an exact majority protocol for trees that uses
  $2^{O(\Delta^2)}$
  states and stabilises in $O(n^2 \log n)$ steps in expectation and with high probability.
\end{restatable}

For paths, this is optimal up to the $\log n$ factor due to the recent $\Omega(Dm)$ time lower bound on general graphs~\cite{rybicki2026space}.
While our protocol is simple, its analysis requires care.
A na\"ive coupon collector analysis yields a bound of $O(n^3 \log n)$ steps.
However, to get to the worst-case optimal $\Theta(n^2 \log n)$ steps, a more careful analysis is needed to take into account ``congestion'' that may happen in the process: tokens of the same type block each other from propagating towards the root.

Our technique for analysing the directed annihilation dynamics on trees is inspired by the analysis of the fun-sort algorithm of Biedl et al.~\cite{biedl_funsort_2004}. The fun-sort algorithm corresponds to a simpler process with only two token types on a path, where
the tokens do not annihilate each other, but rather traverse in opposite directions on a~path.

Our analysis for directed trees instead defines a partial order on the input tokens that captures how tokens of the same type can block each other as they move towards the root. This partial order is then used to show that some token must either move at constant speed towards the root or be cancelled.

\subsubsection{Further example applications}

The self-stabilising tree orientation algorithm can be used to solve other problems where information needs to be aggregated to the root and/or information needs to be disseminated from the root. Here we sketch two example applications.
The time-optimality of the protocols follows from a straightforward application of our lower-bound lemma in \Cref{sec:lower-bound-lemma}.

\subparagraph{Time-optimal self-stabilising 2-colouring of bounded-degree trees.}
Given the tree orientation protocol, it is easy to design a self-stabilising, time-optimal 2-colouring algorithm for bounded-degree trees:
each node stores an output bit and whenever it interacts with its parent, it copies and flips the output bit of its parent. The time complexity follows from the combination of  \Cref{coro:orientation} and the time for the parent to broadcast its bit to all other nodes; the latter is known to take $\Theta(Dn)$ steps in expectation and with high probability~\cite{alistarh2025near}. %

\subparagraph{Exact population size counting in trees.}
Counting the population size is a fundamental problem in population protocols~\cite{izumi2014space,beauquier2015space,aspnes2016time,berenbrink2019counting}.
In trees, we can now obtain a simple protocol for
counting the exact size of the tree: each node starts with a token ``1''. The tokens are pushed towards the root. When tokens with values $x$ and $y$ meet, the parent's token is replaced with the value $x+y$ and the child retains a token with value $0$.
In parallel, a broadcasting process is run, where the maximum observed token value is broadcast by all the nodes.
It is easy to check that even while the tree orientation algorithm is stabilising, the sum of all token values remains $n$.
After the tree has been oriented, the counting protocol stabilises in $\Theta(\max \{ D, \log n\} \cdot n)$ steps in expectation and with high probability~\cite{alistarh2025near}. This is time-optimal.
The protocol uses $O(n)$ additional states on top of the tree orientation algorithm. In the class of bounded-degree trees, this protocol is both space- and time-optimal.

\section{Related work}\label{sec:related}

\Cref{table:summary} summarises prior work and our new results.
For leader election and exact majority, the prior protocols work in general connected graphs rather than just trees; the table reports the resulting bounds for trees.
Below we compare our results to prior work in more detail.

\subparagraph{Self-stabilising 2-hop colouring.}
Angluin et al.~\cite{angluin_selfstabilising_2008} gave the first population protocols for 2-hop colouring. Our algorithm is inspired by their approach, where each agent memorises a one-bit stamp for each of the~$\Theta(\Delta^2)$ colours; this results in $2^{O(\Delta^2)}$ space complexity. However, they do not provide time complexity bounds under the stochastic scheduler. Indeed, establishing such bounds seems like a challenging task, and we suspect that their protocol is significantly slower than ours in the absence of any ``stamp clearing'' mechanism.

In contrast, the randomised protocol by Sudo et al.~\cite{Sudo2018_Loosely} uses only a single stamp per node. However, it relies on a large colour palette with a dependency on~$N$, where~$N$ is a known upper bound on~$n$, leading to a $\Theta(N^4)$ state complexity. They show that their protocol stabilises in $O(m n^2 \log n)$ time.
Recently, Kanaya et al.~\cite{Kanaya2024_Almost} obtained a better time complexity of~$O(mn)$ steps by using even more states: their randomised algorithm employs a colour palette of size~$\Theta(N^3 \Delta^2)$ and requires memorising~$\Delta$ colours during the execution, resulting in $N^{O(\Delta)}$ states. Their slower, but deterministic, algorithm requires~$2^{O(N)}$ states and is not mentioned in \Cref{table:summary}.

Compared to these results, when~$\Delta$ is small, our new protocol achieves significantly better state complexity and faster stabilisation time with a smaller colour palette of size $O(\Delta^2)$. %
Our analysis is also quite different: we use a stochastic drift argument~\cite{kotzing_theory_2024} to bound the stabilisation time of the 2-hop colouring protocol. We note that
Bertschinger et al.~\cite{Bertschinger2020} also used drift analysis to bound the stabilisation time of a node colouring process in a different asynchronous model, where activated nodes can see the colours in their entire neighbourhood.

\begin{table}[t]
  \centering
  \renewcommand{\arraystretch}{1.3}
  \begin{tabular}{@{}lllll@{}}
    \toprule
    Graph class  & Problem & States & Expected time & References \\ \midrule
    General & 2-hop col. & $2^{O(\Delta^2)}$ & No bound given. & Angluin et al.~\cite{angluin_selfstabilising_2008} \\
    & & $O(n^4)$ & $O(m n^2 \log n)$ & Sudo et al.~\cite{Sudo2018_Loosely} \\
    &  & $n^{O(\Delta)}$ & $O(mn)$ & Kanaya et al.~\cite{Kanaya2024_Almost} \\
    &  & $2^{O(\Delta^2)}$ & $O(\Delta m \log n)$ & \textbf{This paper:} \Cref{sec:2-hop} \\
    &  & Any & $\Omega(m \log n)$ & \textbf{This paper:} \Cref{sec:2-hop_lower-bound}  ($\dagger$) \\
    \midrule
    Trees & Orientation  & $2^{O(\Delta^2)}$ & $O(Dn+\Delta n \log n)$ & \textbf{This paper:} \Cref{sec:orient} \\
    &      & Any & $\Omega(Dn)$ & \textbf{This paper:} \Cref{sec:lowerbounds}  ($\dagger$) \\
    \cmidrule{2-5}
    & LE   & $6$ & $O(D n^2 \log n)$ & Beauquier et al.~\cite{beauquier2013self}  ($\star$) \\
    &      & $O(\log^2 n)$ & $O( ( D +\log n) n \log n)$ & Alistarh et al.~\cite{alistarh2025near} \\
    &      & $2^{O(\Delta^2)}$ & $O(Dn+\Delta n \log n)$  & \textbf{This paper:} \Cref{sec:tree-le} \\
    &      & Any & $\Omega(Dn)$ & \textbf{This paper:} \Cref{sec:lowerbounds} ($\dagger$) \\
    \cmidrule{2-5}
    & EM   & $4$ & $O(Dn^2 \log n)$ & Bénézit et al.~\cite{benezit2009interval} ($\star$)\\
    &      & $O(\log^2 n)$ & $O(\Delta Dn^2 \log^2 n)$ & Rybicki et al.~\cite{rybicki2026space} \\
    &      & $2^{O(\Delta^2)}$ & $O(n^2 \log n)$ & \textbf{This paper:} \Cref{sec:maj} \\
    &      & Any & $\Omega(Dn)$ & Rybicki et al.~\cite{rybicki2026space} ($\dagger$) \\
    \bottomrule
  \end{tabular}
  \caption{Space and time complexities of population protocols in different graph families.
    Here, $\Delta$ is the maximum degree and $D$ is the diameter of the graph.
    ($\dagger$) This is an unconditional time lower bound for protocols that use any number of states.
  ($\star$) The bounds on the time complexity for the protocols are not given in the reference but follows easily from standard hitting time arguments~\cite{sudo2021self,alistarh2025near}.}
  \label{table:summary}
  \vspace{-1em}
\end{table}

\subparagraph{Tree orientation and extensions to general graphs.}

We are not aware that the tree orientation problem has been previously studied in the context of population protocols. However, we note that Angluin et al.~\cite{angluin_selfstabilising_2008} gave a self-stabilising protocol that finds a spanning tree in 2-hop coloured regular graphs of diameter $D$ using $\poly(D)$ states. Thus, our protocols for directed trees  could also be applied in more general graphs by running them in the resulting spanning tree, with the additional $\poly(D)$ state complexity overhead.

\vspace{-0.5em}
\subparagraph{Leader election on graphs.}
Chen and Chen~\cite{chen2019self}  gave a constant-space \emph{self-stabilising} leader election protocol for cycles with an exponential stabilisation time. Later, they gave a $\poly(\Delta)$-state protocol for general $\Delta$-regular graphs~\cite{chen2020self}.
Yokota et al.~\cite{yokota2021time} gave a time-optimal algorithm for self-stabilising leader election on cycles that stabilises in $\Theta(n^2)$ steps and uses $O(n)$ states.
Recently, Yokota et al.~\cite{yokota2023near} gave a near-time-optimal algorithm that stabilises in $O(n^2 \log n)$ steps and uses only $\polylog n$  states.

Unlike for regular graphs, Angluin et al.~\cite{angluin_selfstabilising_2008} showed that self-stabilising leader election cannot be solved in non-simple graph classes. A class of graphs $\mathcal{G}$ is non-simple if it contains a graph with two disjoint subgraphs that are also contained in $\mathcal{G}$. For example, trees are non-simple. Our protocol for trees is ``almost'' self-stabilising in the sense that as long as the initial configuration contains at least one leader initially (but otherwise the initial states may be arbitrary), then some node gets elected as a leader.

Beauquier et al.~\cite{beauquier2013self}
gave a (non-self-stabilising) constant-state leader election protocol that stabilises on any connected graph.
Using hitting time arguments, this protocol can be shown to stabilise in $O(n^4 \log n)$ time on any graph and in $O(n^3 \log n)$ time on regular graphs~\cite{alistarh2025near,sudo2021self}.
For trees, the protocol stabilises in $O(Dn^2 \log n)$ time, as the hitting time of the population random walk in a diameter-$D$ tree is $O(Dn^2)$.
The (non-self-stabilising) leader election protocol of Alistarh et al.~\cite{alistarh2025near} uses $O(\log^2 n)$ states and stabilises in $O((D+\log n) m \log n)$ steps on any graph, and on graphs with edge expansion $\beta>0$, the protocol stabilises in $O(m/\beta \log^2 n)$ time. For trees, their protocol achieves $O((D+\log n) n \log n)$ stabilisation time.
In comparison, our protocol is strictly better in terms of space and time complexity when~$\Delta$ is small, and strictly worse when~$\Delta$ is large, e.g., $\Delta \gg \sqrt{\log \log n}$.

\subparagraph{Exact majority on graphs.}
There is a simple 4-state protocol that solves exact majority on any connected graph~\cite{benezit2009interval,draief2012convergence,mertzios2017determining}.
Mertzios et al.~\cite{mertzios2017determining} showed that there is no always-correct 3-state protocol.
Draief and Vojnovi\'{c}~\cite{draief2012convergence} studied the stabilisation time of the 4-state protocol using spectral techniques,
and recently, Rybicki et al.~\cite{rybicki2026space} extended their analysis  to give a bound of $O((\tau/\gamma) \cdot \log n)$ expected steps, where $\tau$ is the relaxation time of the \emph{population random walk} and $1/n < \gamma < 1$ is the normalised bias of the inputs.%

In trees, the relaxation time $\tau$ of the population random walk can range from $O(n)$ (e.g., stars) to $\Theta(n^3)$ (e.g., paths).
However, plugging in the worst-case values for $\tau$ and $\gamma$ leads to too pessimistic bounds.
Again, with simple hitting time arguments, such as those used by Sudo et al.~\cite{sudo2021self} and Alistarh et al.~\cite{alistarh2025near}, one can show that the 4-state protocol stabilises in $O(Dn^2 \log n)$ steps in any diameter-$D$ tree. %

Rybicki et al.~\cite{rybicki2026space} also gave an $O(\log n \cdot (\log(\Delta/\delta) + \log(\tau / n)))$-state  protocol that stabilises in $O((\tau\Delta/\delta) \log^2 n)$ steps, where $\delta$ is the minimum degree and $\Delta$ is the maximum degree of the input graph. The protocol performs well on almost-regular graphs with good expansion.
However, in worst-case graphs, such as lollipop graphs, the relaxation time $\tau$ can be $\Omega(n^4)$, making the protocol slow.
Moreover, the clocking mechanism performs poorly in graphs where $\Delta/\delta$ is large.
On trees of diameter $D$, one can use the bound  $\tau \in O(Dn^2)$ to get a running time bound of $O(\Delta D n^2 \log^2 n) = O(n^4 \log^2 n)$.
However, we suspect that this bound could be improved by a $\Theta(n)$-factor with a more careful analysis of the underlying cancellation-doubling dynamics on undirected trees.

In contrast, our new majority protocol for trees stabilises in only $O(n^2 \log n)$ steps  on any tree. %
For bounded-degree trees, our new protocol uses $O(1)$ states while achieving a linear speed-up compared to both the 4-state protocol and the super-constant state protocol in the worst case.
This is almost optimal, as any exact majority protocol requires $\Omega(n^2)$ expected time on paths~\cite{rybicki2026space}.
However, the space complexity of our protocol scales worse in terms of $\Delta$.

\section{Fast 2-hop colouring in general graphs}\label{sec:2-hop}

In this section, we give a self-stabilising algorithm whose output stabilises to a 2-hop colouring with $O(\Delta^2)$ colours on any graph $G = (V,E)$.
The algorithm is inspired by the non-deterministic 2-hop colouring algorithm of Angluin et al.~\cite{angluin_selfstabilising_2008} that uses stamping to detect conflicts. However, to ensure fast stabilisation under the stochastic scheduler, our algorithm is different in its construction and assumes that nodes can sample $O(\log \Delta)$ i.i.d.\ random bits per interaction. We prove the following:

\twohop*

\subsection{The 2-hop colouring algorithm} \label{sec:2-hop_colouring_algo}

Fix a constant $\alpha>6$ and let $\C := \{1,\ldots,\alpha \Delta^2\}$ be our colour palette. The algorithm maintains a colour variable $\col(u) \in \C$ and a stamp vector $\stamp(u) \in \{0,1,\bot\}^\C$. The output is given by the map $\col \colon V \to \C$. The $\stamp$ variables are used to check for the validity of the colouring.
Specifically, we say the edge $e$ has a {\em stamp conflict} at time $t$ if
\begin{align*}
  &\bot \notin \{\stamp_t(u, \col_t(v)), \stamp_t(v,\col_t(u)) \}, \text{ and } \\
  &\stamp_t(u, \col_t(v)) \neq \stamp_t(v, \col_t(u)).
\end{align*}
In words, edges have a stamp conflict if the stamps corresponding to their respective colours do not match \emph{and} neither of the stamps is a $\bot$-stamp; the $\bot$-stamp acts as a special stamp to indicate that the stamps were recently cleared.
The algorithm is given in the pseudo-code below and \Cref{figure:definition_conflicts} illustrates stamp conflicts.

\begin{algorithm} [htbp]
  \DontPrintSemicolon
  \If{the activated edge $\{u,v\}$ has a stamp conflict}{
    Sample $\col(u) \gets \uniform(\C)$. \;
    Sample $\col(v) \gets \uniform(\C)$. \;
    Clear stamps by setting $\stamp(v) \gets \{\bot\}^\C$ and $\stamp(u) \gets \{\bot\}^\C$.
  }
  Sample $X \gets  \uniform(\{0,1\})$. \;
  Set $\stamp(u, \col(v)) \gets X$. \;
  Set $\stamp(v, \col(u)) \gets X$. \;
  \caption{{\bf \colouringalgo}}
  \label{alg:2-hop-colouring}
\end{algorithm}

\subparagraph{Conflict edges and recolouring steps.}
For the analysis, we need a few additional definitions.
We write $N(v)$ for the neighbourhood of $v$.
A path $(u,v,w)$ of length 2 is a \emph{$c$-conflict path at time $t$} if $\col_t(u) = \col_t(w) = c$.
We say that an edge $e = \{u,v\}$ has a \emph{colour conflict at time $t$} if
there exists~$w \in N(v) \setminus \{u\}$ such that~$\col_t(w) = \col_t(u)$.
Colour conflicts correspond to edges where the condition for having a 2-hop colouring is not satisfied (i.e., one node has at least two neighbours with the same colour). In particular, an edge with a colour conflict lies on at least one conflict path; see \Cref{figure:definition_conflicts} for an illustration.

\begin{figure}[htbp]
  \centering
  \includegraphics[width=0.4\linewidth]{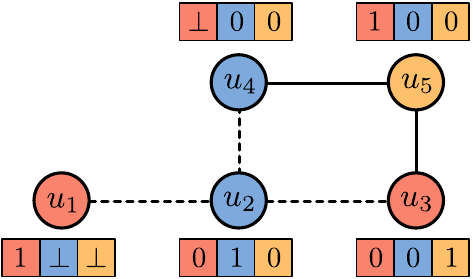}
  \caption{Different types of edge conflicts. The array indexed by colours next to each node represents the $\stamp$ variable of the corresponding node. In this example, the only edge with a stamp conflict is $\{u_2,u_4\}$. In addition, $\{u_1,u_2\}$ and $\{u_2,u_3\}$ both have a colour conflict, and $(u_1,u_2,u_3)$ is a ``red''-conflict path. Finally, $\{u_3,u_5\}$ and $\{u_4,u_5\}$ do not have any conflict.}
  \label{figure:definition_conflicts}
\end{figure}

If the edge $e_t$ sampled at time $t$ has a stamp conflict, then the \colouringalgo triggers a recolouring, i.e., it enters the `if' statement and changes the colour of both activated nodes.
We call such steps \emph{recolouring steps} and write~$\R_t$ to denote the event: ``Time step~$t$ is a recolouring~step''.

We say that $e$ is a \emph{conflict edge} at time $t$ if it has a colour conflict or a stamp conflict at time $t$.
Note that a stamp conflict does not imply a colour conflict, and vice versa.
A recolouring step can only happen when an edge with a stamp conflict is sampled.
Note also that after a recolouring step, all the stamps of $u$ and $v$ are cleared to the value $\bot$.
However, after a recolouring step, we may still have that the colours of $u$ and $v$ remain the same.

We will analyse the stabilisation time of the protocol using stochastic drift analysis on the expected number of conflict edges.
First, we analyse the effect of a single interaction step, and then we consider how the number of conflict edges decreases in a sequence of $m$ interaction steps.

\subsection{Effect of a single interaction step}

We now analyse how the number of conflict edges changes during the execution of the protocol.
First, we show that the set $C_t$ of conflict edges remains unchanged after any step that is \emph{not} a recolouring step.
This implies that the number of conflict edges can only change during recolouring~steps.

\begin{lemma} \label{lemma:no_recolouring}
  If~$t$ is not a recolouring step, then~$C_t = C_{t+1}$.
\end{lemma}
\begin{proof}
  Suppose that $t$ is not a recolouring step.
  Let $e = \{u,v\} \in E$ be any edge.
  First, observe that $e$ has a colour conflict at time~$t$ if and only if it has a colour conflict at time~$t+1$.
  This is because by assumption $t$ is not a recolouring step, and hence, the colours of nodes cannot change in this step.

  Therefore, it remains to consider the case that~$e$ does not have a colour conflict at time~$t$ and~$t+1$.
  We distinguish between two subcases depending on whether $e$ is activated at time $t$ or not:
  \begin{enumerate}%
    \item
      Suppose $e=e_t$ is the activated edge at time~$t$.
      Note that $e$ cannot have a stamp conflict at time~$t$, because otherwise the activation would trigger a recolouring step,
      which is assumed not to be the case.

      The edge $e$ also cannot have a stamp conflict at time~$t+1$ by construction of the \colouringalgo.
      Therefore, $e$ cannot be a conflict edge at time $t$ or $t+1$, i.e.,  $e \notin C_t \cup C_{t+1}$.

    \item
      Suppose $e \neq e_t$ is not the activated edge at time~$t$.
      Note that modifying the variable $\stamp(u,\col_t(v))$ can only be achieved if an edge of the form~$e_t = \{u,w\}$ is activated, with~$\col_t(w) = \col_t(v)$; but this would imply that~$e$ has a colour conflict, which we have already ruled out.
      Therefore, whether $e$ has a stamp conflict (or not) at time~$t$, remains the same at time~$t+1$. \qedhere
  \end{enumerate}
\end{proof}

Next, we show that in each recolouring step, the number of conflict edges  decreases by a positive constant amount in expectation.

\begin{lemma} \label{lemma:pain_in_the_**s}
  There is a constant~$\eta_1 > 0$ such that for every time step~$t \ge 0$, we have
  \begin{equation*}
    \Exp\big[ \ |C_t| - |C_{t+1}| \ \mid \R_t \big] \geq \eta_1.
  \end{equation*}
\end{lemma}

The proof of \Cref{lemma:pain_in_the_**s} is somewhat involved; the rest of this subsection is dedicated to it.
We fix a time step $t$ and assume that $t$ is a recolouring step, i.e., we condition on $\R_t$.
Let~$e = \{u,v\} = e_t$ be the edge sampled to interact at time~$t$.

\subparagraph{Risky paths.}
We say that a length-2 path $(u,w,z)$ is \emph{$c$-risky for node $u$} if $z$ has colour $c$ at time $t+1$.
The idea is that this path may become (or remain) a $c$-conflict path at time $t+1$.
We define $E_u(c)$ to be the set of all \emph{edges} that lie on some $c$-risky path for node $u$.
Formally,
\begin{equation*}
  E_u(c) := \bigcup_{\substack{w,z \in V \setminus \{u\}}} \Big\{ \{u,w\} , \{w,z\} :  \col_{t+1}(z) = c, \{u,w\} \in E, \{w,z\} \in E \Big\}.
\end{equation*}

\begin{figure}[htbp]
  \centering
  \includegraphics[width=0.35\linewidth]{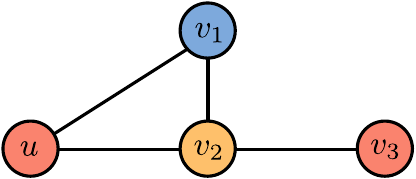}
  \caption{Here, $(u,v_1, v_2)$ is yellow-risky; $(u,v_2, v_1)$ is blue-risky; and $(u,v_2, v_3)$ is red-risky.}

\end{figure}

Let~$X_u,X_v$ be the random colours adopted by~$u$ and~$v$ respectively at time~$t+1$ (after the recolouring).
We now show that~$E_u(X_u)$ contains all the edges that might get a new conflict, as a consequence of~$u$ adopting colour~$X_u$ (and similarly for~$v$).

\begin{claim} \label{claim:conflict_characterization}
  Suppose $e' \notin C_t$ and $e' \in C_{t+1}$.
  Then $e' \in E_u(X_u) \cup E_v(X_v)$.
\end{claim}
\begin{proof}

  Let $e' = \{x,y\}$.
  By assumption, the edge $e' \notin C_t$ has no stamp conflict and no colour conflict at time~$t$.
  \begin{itemize}
    \item If $e'$ is not adjacent to $e_t$, then $x$ and $y$ did not change their states at time~$t$, in particular, $e'$ has no stamp conflict at time~$t+1$.
      However, since $e' \in C_{t+1}$, this means that $e'$ has a colour conflict at time $t+1$.

    \item If $e'$ is adjacent to $e_t$, then we have $\left|e' \cap e_t \right|=1$ because $e' \not\in C_t$ and $e_t \in C_t$ implies $e' \neq e_t$.
      Without loss of generality, let's say that $x = u \in e' \cap e_t$.
      Since $x$ got recoloured at time~$t$, it holds that $\stamp_{t+1}(x,c)=\bot$ for all $c \in \mathcal{C}\setminus\{X_v\}$.
      Therefore, assuming that $e'$ has a stamp conflict at time $t+1$, it must be the case that $\col_{t+1}(y) = X_v$; in particular,
      $e'$ also has a colour conflict at time~$t+1$.
  \end{itemize}
  In either case, we have shown that $e'$ has a colour conflict at time~$t+1$.
  Therefore, there must exist some $c$-conflict path $(x,y,z)$ or $(y,x,z)$ at time $t+1$.
  Without loss of generality, let $(x,y,z)$ be such a path.
  Since~$e'$ has no colour conflict at time~$t$, this implies that either node~$z$ or node~$x$ changed their colour at time~$t$.
  This can only happen if they are incident to the edge $e_t = \{u,v\}$.
  Thus, both endpoints of the path $(x,y,z)$ have the same colour at time $t+1$ and one of these endpoints is $u$ or $v$.
  Since~$e'$ belongs to this path, the edge lies either in $E_u(X_u)$ or~$E_v(X_v)$.
\end{proof}

To later bound the number of edges contained in~$E_u(X_u) \cup E_v(X_v)$, we now establish an intermediate result about the sum of all~$E_u(c)$.
\begin{claim} \label{claim:sum_of_E_u}
  For $w \in \{u,v\} = e_t$, we have that
  \begin{equation} \label{eq:sum_of_E_u}
    \sum_{c \in \C} |E_w(c)| \leq 2 \Delta^2.
  \end{equation}
\end{claim}
\begin{proof}
  By definition, any edge on a $c$-risky path from $w$ is either (i)~incident to node $w$, or (ii)~incident to a neighbour $x$ of $w$, but not incident to $w$. We first bound the number of both types of edges.
  \begin{itemize}
    \item[(i)] Any edge $e'$ that is incident to node $w$ can be a part of at most $\Delta-1$ paths of length~2 from $w$, because $\deg(x) \le \Delta$. In particular,
      there are at most~$\Delta-1$ colours~$c$ such that $e'$ lies on a $c$-risky path from $w$.

    \item[(ii)] Any edge $e'' = \{x,y\}$ such that~$x \in N(w)$ and~$y \neq w$ can be part of at most one $c$-risky path $(w,x,y)$ from $w$, because by definition $(w,x,y)$ is a $c$-risky path if and only if $\col_{t+1}(y)=c$.

  \end{itemize}
  There are at most~$\Delta$ edges of type (i) and at most $\Delta(\Delta-1) \le \Delta^2$
  edges of type (ii). Any other edge is too far from~$w$ and cannot belong to any~$E_w(c)$. Therefore, we get that
  \begin{equation*}
    \sum_{c \in \C} |E_w(c)| \leq \Delta (\Delta - 1) + \Delta^2 \leq 2\Delta^2. \qedhere
  \end{equation*}
\end{proof}

Recall that~$\alpha > 6$ is a constant parameter used to define the size of the colour palette, $|\C| = \alpha \Delta^2$.
We now lower bound the probability that the recolouring at time~$t$ does not create a new colour conflict, as a function of~$\alpha$.
\begin{claim} \label{claim:good_event}
  We have that
  \begin{equation*}
    \Pr[E_u(X_u) = \emptyset \text{ and } E_v(X_v) = \emptyset] \geq 1-\frac{2}{\alpha}.
  \end{equation*}
\end{claim}
\begin{proof}
  There are at most~$\Delta^2$ paths of length~$2$ starting from~$u$. Therefore, there are at most~$\Delta^2$ distinct colours~$c$ such that~$E_u(c) \neq \emptyset$. Since~$X_u$ is uniformly distributed in~$\C$, %
  \begin{equation*}
    \Pr[ E_u(X_u) \neq \emptyset ] \leq \frac{\Delta^2}{|\C|} = \frac{1}{\alpha}.
  \end{equation*}
  As the same inequality holds symmetrically for~$v$,  \Cref{claim:good_event} follows from the union~bound.
\end{proof}

Finally, we are ready to conclude the proof of \Cref{lemma:pain_in_the_**s}.
\begin{proof}[Proof of \Cref{lemma:pain_in_the_**s}]
  Let $e_t = \{u,v\}$ be the edge sampled at time $t$. Because $t$ is a recolouring step, the edge $e_t$ has a stamp conflict at time $t$ (a stamp conflict being necessary to trigger a recolouring).
  Moreover, note that if~$E_u(X_u)  = \emptyset$ and $E_v(X_v) = \emptyset$, then~$e_t$ has no colour conflict at time~$t+1$. In that case, since it cannot have a stamp conflict after recolouring, it has no conflict at all at time~$t+1$.
  Thus, the expected number of edges that \emph{cease} to have a conflict at time~$t+1$ is at least
  \begin{equation*}
    1 \cdot \Pr[E_u(X_u)  = \emptyset \text{ and } E_v(X_v) = \emptyset] \geq 1-\frac{2}{\alpha},
  \end{equation*}
  where the inequality follows from \Cref{claim:good_event}.
  The expected number of \emph{newly created} conflict edges at time $t+1$ is by Claim~
  \ref{claim:sum_of_E_u} and Claim~\ref{claim:conflict_characterization} at most
  \begin{equation*}
    \Exp\Big[|E_u(X_u)| + |E_v(X_v)|\Big] = \sum_{c \in \C} \frac{1}{|\C|} \left(|E_u(c)| + |E_v(c)| \right) \leq \frac{4\Delta^2}{|\C|} = \frac{4}{\alpha}.
  \end{equation*}
  Overall, we obtain
  \begin{equation*}
    \Exp\Big[ |C_t| - |C_{t+1}| \ \mid \R_t \Big] \geq 1-\frac{2}{\alpha}-\frac{4}{\alpha} = 1-\frac{6}{\alpha} =: \eta_1,
  \end{equation*}
  which is bounded away from~$0$ as long as~$\alpha > 6$.
\end{proof}

\subsection{\texorpdfstring{Effect of~$m$ consecutive steps}{Effect of m consecutive steps}}

We now analyse how~$|C_t|$ changes over {\em intervals} of~$m = |E|$ consecutive time steps.
For a given time~$t \geq 0$, let  $I(t) = \{t,\ldots,t+m-1\}$ be the interval of $m$ steps starting from time step $t$.

Define $Y_t$ to be the number of recolouring steps in the phase starting at time~$t$, i.e.,
\begin{equation*}
  Y_t := \sum_{s=t}^{t+m-1} \bbOne\{\R_s\} = \sum_{s \in I(t)} \bbOne\{\R_s\}.
\end{equation*}
To perform drift analysis on~$|C_t|$, we first prove a lower bound on~$\Exp[Y_t \mid |C_t|]$.
Let~$\stconflict \subseteq C_t$ be the set of edges with a stamp conflict and
$\coconflict \subseteq C_t$ be the set of edges with a colour conflict.
For each of these two subsets, we lower bound the probability that an edge is involved in a recolouring step during the interval $I(t)$.
In the following, notation $e$ refers to the mathematical constant whenever it is used inside a numerical expression, and to an edge of the graph everywhere else.

\begin{lemma} \label{lemma: sampling stamp conflict}
  For every $e \in \stconflict$, the probability that~$e$ is recoloured during the interval $I(t)$ is at least~$\frac{1-e^{-1}}{2\Delta}$.
\end{lemma}
\begin{proof}
  Let $F$ be the set of all edges sharing at least one endpoint with~$e$. In particular, we have that $e \in F$ and
  $1 \leq |F| \leq 2 \Delta$.
  Let $\mathcal{E}_1$ be the event that some edge in $F$ gets activated during the interval~$I(t)$.
  Observe that
  \[
    \Pr[\mathcal{E}_1] \ge 1 - \left( 1 - \frac{|F|}{m} \right )^m \ge 1-(1-1/m)^m \geq 1-1/e > 0.
  \]
  Let $\mathcal{E}_2$ be the event that the edge $e$ is the first edge of $F$ that gets activated during the interval $I(t)$.
  Since the scheduler selects edges uniformly at random in each time step, we~get %
  \[
    \Pr[\mathcal{E}_2 \mid \mathcal{E}_1] \ge \frac{1}{|F|} \ge \frac{1}{2\Delta}.
  \]
  We now show that if the event $\mathcal{E}_1 \cap \mathcal{E}_2$ happens, then the edge $e$ gets recoloured during the interval~$I(t)$.
  Suppose the event $\mathcal{E}_1 \cap \mathcal{E}_2$ happens and let~$s \in I(t)$ be the first time step where~$e$ is activated in the interval.
  By assumption, the edge $e$ has a stamp conflict at time $t$.
  Since no adjacent edge of $e$ is activated between time steps~$t$ and~$s$, nodes $u$ and $v$ do not change their state before time $s$, and therefore $e$ still has a stamp conflict at time $s$. By construction of the \colouringalgo, this implies that~$e$ is recoloured at time~$s$.
  Thus, $e$ gets recoloured during the interval $I(t)$ with probability at least
  \[
    \Pr[\mathcal{E}_1 \cap \mathcal{E}_2] = \Pr[\mathcal{E}_2 \mid \mathcal{E}_1] \cdot \Pr[\mathcal{E}_1]  \ge   \frac{1-e^{-1}}{2\Delta}. \qedhere
  \]
\end{proof}

\begin{lemma} \label{lemma: sampling colour conflict}
  Let $(u,v,w)$ be a $c$-conflict path at time $t$, for some~$c \in \C$.
  Then at least one node in $\{u,v,w\}$ is recoloured during the interval $I(t)$ with probability at least $e^{-2}/16$ for large enough $m$.
\end{lemma}
\begin{proof}
  Let $e = \{u,v\}$ and $e' = \{v,w\}$.
  The number of times that an edge from the set $\{e,e\ap\}$ gets activated during the interval $I(t)$ of length $m$ is a binomial random variable $S \sim \text{Binomial}(m,2/m)$.
  For large enough~$m$, we have
  \begin{align*}
    \Pr[S \geq 3] &\geq \Pr[S=3] = \binom{m}{3} \pa{\frac{2}{m}}^3 \pa{1-\frac{2}{m}}^{m-3} \\
    &= \frac{m(m-1)(m-2)}{6} \cdot \frac{8}{m^3} \cdot \pa{1-\frac{2}{m}}^{m-3} \\
    &\geq \frac{4}{3} \cdot \pa{\frac{m-2}{m}}^3 \cdot \pa{1-\frac{2}{m}}^{m-3} \\
    &\geq \frac{4}{3} \cdot \pa{1-\frac{2}{m}}^{m} \geq e^{-2}, %
  \end{align*}
  where for the last inequality we used the fact that~$\lim_{m \to +\infty} ~ (1-2/m)^{m} = e^{-2}$.
  Moreover, conditioning on~$S \geq 3$, the probability that the first three activations from the set $\{e,e\ap\}$ follow the order $(e,e\ap,e)$ is $1/8$, because each specific order has the same probability of occurring.
  From now on, we condition on this event.

  Let~$s \in I(t)$ be the time step corresponding to the second activation of~$e$ within the phase.
  Consider the case that none of the nodes in $\{u,v,w\}$ got recoloured before time~$s$.
  Then, the colours of $u,v$, and $w$ are still the same at time~$s$ as at time~$t$.
  Furthermore, since $e$ is sampled at least once between times~$t$ and~$s$, and $u$ and $v$ are not recoloured in this interval, we have that $\bot \not\in \{\stamp(u,\col(v)), \stamp(v,\col(u))\}$ at time~$s$.
  We also have that $\stamp(u,\col(v))$ is independent of $\stamp(v,\col(u))$, since there was at least one edge sampled in between (namely $e\ap$), which sets $\stamp(v,\col(w)) = \stamp(v,\col(u))$ uniformly at random and independently of $\stamp(u,\col(v))$.
  Therefore, $s$ is a recolouring step with probability~$1/2$.
  Overall, either~$u$, $v$, or~$w$ gets recoloured within the phase with probability at least
  \begin{equation*}
    e^{-2} \cdot \frac{1}{8} \cdot \frac{1}{2} = \frac{e^{-2}}{16}. \qedhere
  \end{equation*}
\end{proof}

\subparagraph{Path extension functions.}
\Cref{lemma: sampling colour conflict} implies that, with high probability, a large number of conflict paths are partially recoloured during $I(t)$. However, our goal is to obtain a lower bound on the number of recolouring steps. To deduce this from the former statement, we show that conflict paths can be grouped together in a structured way, formalised as follows.

By definition, any edge $e \in \coconflict$ is part of at least one conflict path at time $t$.
We say that $f \colon \coconflict \to \coconflict$ is a \emph{path extension function} if for every~$e \in \coconflict$, the nodes in $e \cup f(e)$ form a conflict path at time~$t$.
The idea is that we can use $f$ to extend an edge $e \in \coconflict$ with a colour conflict to some conflict path that contains $e$.
Moreover, we say that $f$ is \emph{balanced} if $|f^{-1}(e)| \le 2$, i.e., at most two distinct edges can have $e$ as their image.

\begin{lemma}\label{lemma:balanced-paths-exist}
  For any $t \ge 0$, there exists a balanced path extension function.
\end{lemma}
\begin{proof}
  Orient each edge $\{u,v\} \in \coconflict$ such that $\{u,v \}$ is oriented from $u$ towards $v$ if there exists some conflict path  $(u,v,w)$, breaking ties arbitrarily. Clearly, all edges with a colour conflict can be oriented in this manner. For any node $v \in V$ that has $k \ge 2$ neighbours  $z_0, \ldots, z_{k-1}$ of colour $c \in \C$ at time $t$, define
  \[
    \pi_{v,c}(z_i) := z_{i +1 \bmod k}.
  \]
  Note that $\pi_{v,c}$ is a permutation on the $c$-coloured neighbours $z_0, \ldots, z_{k-1}$ of node $v$ and the permutation has no fixed points.
  For any edge $\{u,v\} \in \coconflict$ oriented as $(u,v)$ and $\col_t(u) = c$, we set
  \[
    f(\{u,v\}) := \{ v, \pi_{v,c}(u) \}.
  \]
  That is, any colour conflict edge oriented as $(u,v)$ is mapped to the ``next'' adjacent edge $\{v,w\}$ of $v$ which also has an endpoint with the same colour $c$ as $u \neq w$.
  In particular, now $e \cup f(e) = \{u,v,w\}$ is a conflict path, and hence $f$ is a path extension function.

  It remains to show that $f$ is balanced, i.e., at most two edges get mapped to any given edge.
  Consider any $\{x,y\} \in \coconflict$, where the end points have colours $c_x$ and $c_y$ at time $t$.
  By definition of $f$, only edges oriented as $(u,x)$ and $(v,y)$, where $\col_t(u) = c_y$ and $\col_t(v) = c_x$, can be mapped to $\{x,y\}$. Since the maps $\pi_{x,c_y}$ and $\pi_{y,c_x}$, when defined, are permutations, this implies that at most one incident edge of $x$ with colour $c_y$ gets mapped to $\{x,y\}$ and at most one incident edge of $y$ with colour $c_x$ gets mapped to $\{x,y\}$.
\end{proof}

\subparagraph{Bounding the number of recolouring steps.}
We are now ready to lower bound the expected number $\Exp[Y_t \mid |C_t|]$ of recolouring steps that occur during the interval~$I(t)$.

\begin{lemma} \label{lemma: bound on Y}
  There is a constant~$\eta_2 > 0$ such that for all $t \ge 0$ we have
  \[
    \Exp\Big[Y_t ~\big|~ |C_t|\Big] \geq \eta_2 \cdot |C_t| \cdot \Delta^{-1}.
  \]
\end{lemma}

\begin{proof}
  Note that~$C_t = \stconflict \cup \coconflict$, and therefore, $|C_t| \leq |\stconflict| + |\coconflict|$. Hence, we have that~$|\coconflict| \geq |C_t|/2$ or $|\stconflict| \geq |C_t|/2$; we will analyse both cases separately:

  \begin{enumerate}
    \item First, suppose that $|\stconflict| \geq |C_t|/2$. Define~$Z_e \in \{0,1\}$ as an indicator variable for the event that the stamp conflict edge $e \in \stconflict$ gets recoloured during the interval $I(t)$.
      By definition of $Y_t$, we have
      \[
        Y_t = \sum_{s \in I(t)} \bbOne\{\R_s\} \ge \sum_{e \in \stconflict} Z_e.
      \]
      By linearity of conditional expectation and \Cref{lemma: sampling stamp conflict}, we obtain
      \begin{equation*}
        \Exp\Big[Y_t ~\big|~ |C_t|\Big] \geq \sum_{e \in \stconflict} \Exp\Big[Z_e ~\big|~ |C_t|\Big] \geq \frac{1-e^{-1}}{2\Delta} \cdot |\stconflict| \geq \frac{1-e^{-1}}{4\Delta} \cdot |C_t|.
      \end{equation*}

    \item For the second case, suppose  that~$|\coconflict| \geq |C_t|/2$.
      Let $f \colon \coconflict \to \coconflict$ be a balanced path extension function given by \Cref{lemma:balanced-paths-exist} at time $t$. For each  $e \in \coconflict$, let $W_e \in \{0,1\}$ be the indicator variable for the event that at least one node in $e \cup f(e)$ is recoloured during the interval $I(t)$.
      Note that $f(e) \cup e = \{u,v,w\}$ forms a conflict path, since $f$ is a path extension function.
      Define
      \[
        H:=\sum_{e \in \coconflict} W_e.
      \]
      Since $f(e) \cup e = \{u,v,w\}$ forms a conflict path, \Cref{lemma: sampling colour conflict} gives  $\Exp[W_e \mid |C_t|] \ge e^{-2}/16$.
      Applying linearity of conditional expectation yields
      \begin{equation} \label{eq:H_expectation}
        \Exp\Big[H ~\big|~ |C_t|\Big] = \sum_{e \in \coconflict} \Exp\Big[W_e ~\big|~ |C_t|\Big] \geq \frac{e^{-2}}{16} \cdot |\coconflict| \geq \frac{e^{-2}}{32} \cdot |C_t|,
      \end{equation}
      where the last inequality follows from our assumption that $|\coconflict| \geq |C_t|/2$.

      Next, we show that $Y_t \geq H/(6\Delta)$ by bounding the number of indicator variables that can be set to 1 in a single recolouring step.
      Consider the case that~$s$ is a recolouring step and $\{u,v\}$ is the edge recoloured at time~$s$.
      Let~$F$ be the set of all edges in $\coconflict$ incident to~$u$ or~$v$:
      \begin{equation*}
        F := \{e \in \coconflict :  u \in e \text{ or } v \in e \}.
      \end{equation*}
      We have that $|F| \leq 2 \Delta$. Let~$F'$ be the set of all edges in $\coconflict$ that are mapped by $f$ to an edge of $F$:
      \begin{equation*}
        F' := \{e \in \coconflict : f(e) \in F\}.
      \end{equation*}
      Since $f$ is balanced, the size of $F'$ is bounded by $|F'| \leq 2 |F| \leq 4 \Delta$.
      Now, let~$e \in \coconflict \setminus (F \cup F')$.
      We have that $\{u,v\} \cap e = \emptyset$ (otherwise $e$ would belong to $F$) and $\{u,v\} \cap f(e) = \emptyset$ (otherwise $f(e)$ would belong to $F$). Therefore, $\{u,v\} \cap (e \cup f(e)) = \emptyset$, and $W_e$ is not set to $1$ at time~$s$. Overall, a single recolouring step can only set at most~$|F \cup F'| \leq 6\Delta$ indicator variables to~$1$.
      This implies that $Y_t \geq H/(6\Delta)$, which together with \Cref{eq:H_expectation}, yields
      \[
        \Exp\Big[Y_t ~\big|~ |C_t|\Big] \geq \frac{\Exp\Big[H ~\big|~ |C_t|\Big]}{6 \Delta} \geq \frac{e^{-2}}{192 \, \Delta} \cdot |C_t|.
      \]
  \end{enumerate}
  Taking~$\eta_2 := e^{-2}/192 < \frac{1-e^{-1}}{4}$ concludes the proof of \Cref{lemma: bound on Y}.
\end{proof}

\twohop*
\begin{proof}
  The bound on the number of states immediately follows from the definition of the algorithm.
  \Cref{lemma:no_recolouring,lemma:pain_in_the_**s} imply that for any $t \ge 0$ we have
  \begin{equation} \label{eq:single_step_expectation}
    \begin{split}
      \Exp\Big[ |C_t| - |C_{t+1}| ~\big|~ |C_t| \Big] &= \Pr[\R_t] \cdot \Exp\Big[ \ |C_t| - |C_{t+1}| \ \mid \R_t \Big]  + \Pr[\overline{\R_t}] \cdot \Exp\Big[ \ |C_t| - |C_{t+1}| \ \mid \overline{\R_t} \Big] \\
      &\geq \eta_1 \cdot \Pr[\R_t].
    \end{split}
  \end{equation}
  Using \Cref{lemma: bound on Y}, we get that the expected decrease in the number of conflict edges  in a single~phase~is
  \begin{align*}
    \Exp\Big[ |C_t| - |C_{t+m}| ~\big|~ |C_t| \Big] &= \sum_{s=t}^{t+m-1} \Exp\Big[ |C_s| - |C_{s+1}| ~\big|~ |C_t| \Big]  \\
    &\geq \sum_{s=t}^{t+m-1} \eta_1 \cdot \Pr\Big[\R_s ~\big|~ |C_t|\Big] = \eta_1 \cdot \Exp\Big[Y_t ~\big|~ |C_t|\Big] & \text{(by \Cref{eq:single_step_expectation})} \\
    &\geq \frac{\eta_1 \, \eta_2}{\Delta} \cdot |C_t|. & \text{(by \Cref{lemma: bound on Y})}
  \end{align*}
  By the multiplicative drift theorem (\Cref{thm:multiplicative_drift} in Appendix~\ref{apx:tools}) applied to~$X_k := |C_{km}|$, we have that $K = \inf \{k : X_k = 0 \}$ is $O(\Delta \log m) = O(\Delta \log n)$ with high probability.
  By definition, $|C_{Km}| = 0$. The absence of conflict edges implies that the $\col$ variables form a 2-hop colouring of~$G$ at time $Km$.
  Furthermore, as there are no more stamp conflicts, the colouring is stable at time~$Km$.

  This implies that for $T \in O(\Delta \, m \log n)$ we have that $|C_t| = 0$ for all $t \ge T$ with high probability.
  Thus, the protocol stabilises in $O(\Delta \, m \log n)$ steps with high probability.
  The expected time bound follows easily from the tail bound; see \Cref{lemma:whp-to-exp} in the Appendix.
\end{proof}

\subsection{Time lower bound for 2-hop colouring}\label{sec:2-hop_lower-bound}

We now give a lower bound for 2-hop colouring that follows easily from the coupon collector problem.

\begin{theorem} \label{thm: 2-hop lower bound}
  There is no protocol for 2-hop colouring that stabilises in $o(m \log n)$ expected steps.
\end{theorem}
\begin{proof}
  For~$m \in [2n,n^2]$, let $H$ be any connected graph with $n$ nodes and $m-n$ edges, and $P$ be a path graph with $n$ nodes $v_1,\dots,v_n$.
  We define $G$ to be the union of $H$ and $P$, connected by a single edge between $v_1$ and an arbitrary node of $H$.
  By construction, $G$ has $2n$ nodes and $m$ edges.
  We show that any protocol, using possibly an infinite number of states and random bits, needs $\Omega(m \log n)$ expected steps to stabilise on $G$.

  For this, we choose pairwise disjoint sets of edges $E_1,\dots,E_s$ with $s:=  \lfloor (n-1)/4 \rfloor$ and
  $E_i := \{\{v_j,v_{j+1}\} : j = 4i-3,\dots,4i\}$, i.e., each $E_i$ contains 4 consecutive edges of $P$.
  We say that $E_i$ is \textit{activated} at time $t$ if the sampled edge $e_t$ belongs to $E_i$.
  We assume all nodes start in the same initial state.
  First, we show that every $E_i$ needs to be activated at least once for any protocol to successfully stabilise.
  Without loss of generality, assume that $E_1 = \{\{v_1,v_2\},\dots,\{v_4,v_5\}\}$ has not been
  activated yet, in particular, no edge of $E_1$ has been sampled.
  Then the nodes $v_2$ and $v_4$ are still in their initial state.
  Moreover, as all nodes are initialised in the same state, the nodes $v_2$ and $v_4$ output
  the same colour, in particular, $\{v_2,v_3\}$ and $\{v_3,v_4\}$ both have a colour conflict. Therefore, the current configuration is not correct and
  the protocol cannot have successfully stabilised.

  Now we lower bound the expected number of steps until every $E_i$ has been activated at least once.
  By the coupon collector problem,
  the scheduler needs to sample an edge from $\bigcup_{i=1,\dots,s} E_i$ an expected number of $\Theta(s \log s)$ times before we
  have activated each set $E_i$ at least once.
  Additionally, sampling an edge from $\bigcup_{i=1,\dots,s} E_i$ takes $\Omega(m/s)$ expected steps.
  Therefore, since $s \in \Theta(n)$, the expected number of steps until every $E_i$ has been activated at least once
  is
  \[
    \Omega\Big(\frac{m}{s} \cdot s \log s\Big) = \Omega(m \log n). \qedhere
  \]
\end{proof}

It is easy to see that the lower bound is tight for protocols that have access to sufficiently many random bits: the nodes can simply sample a random colour from the set $\{1, \ldots, \poly(n)\}$. In that case, with high probability, recolouring steps assign unique colours to nodes, and the protocol may stabilise after sampling each edge only a constant number of times.
\section{Time-optimal self-stabilising tree orientation}\label{sec:orient}

In this section, we consider the problem of orienting an undirected tree, assuming we have a 2-hop colouring.
Formally, we show the following result.

\orient*

\subsection{The tree orientation algorithm}

We now assume that the interaction graph $G = (V,E)$ is a tree.
Let $\col \colon V \to \C$ be a stable 2-hop colouring of $G$, where $\C$ is a set of $\chi$ colours.
The 2-hop colouring can be obtained using the self-stabilising algorithm given in \Cref{sec:2-hop}.
In the tree orientation protocol, each node $v$ stores a value $\parent(v) \in \C \cup \{ \bot \}$ acting as a parent pointer, and a set $\children(v) \subseteq \C$ containing the colours of neighbours it considers its children.

\subparagraph{Orientation of an edge.}
The algorithm and its analysis rely on the following three \emph{orientation statuses} of edges; see \Cref{fig:definition} for illustration.
Let $e = \{u,v\}$ be an edge and~$t \ge 0$.
\begin{enumerate}
  \item We say $e$ is \emph{oriented} from $u$ to $v$ at time $t$ if the following conditions are satisfied:
    \[
      \parent_t(v) \neq \col(u) \qquad \col(v) \notin \children_t(u) \qquad  \col(u) \in \children_t(v).
    \]

  \item An edge $e$ that is oriented from $u$ to $v$ is \emph{properly oriented} if also $\parent_t(u) = \col(v)$ is satisfied. Otherwise, we say that $e$ is \emph{weakly oriented} from $u$ to $v$.

  \item If $e$ is not (weakly or properly) oriented, then $e$ is \emph{disoriented}.
\end{enumerate}

\begin{figure}[ht] %
  \centering
  \includegraphics[width=0.89\textwidth]{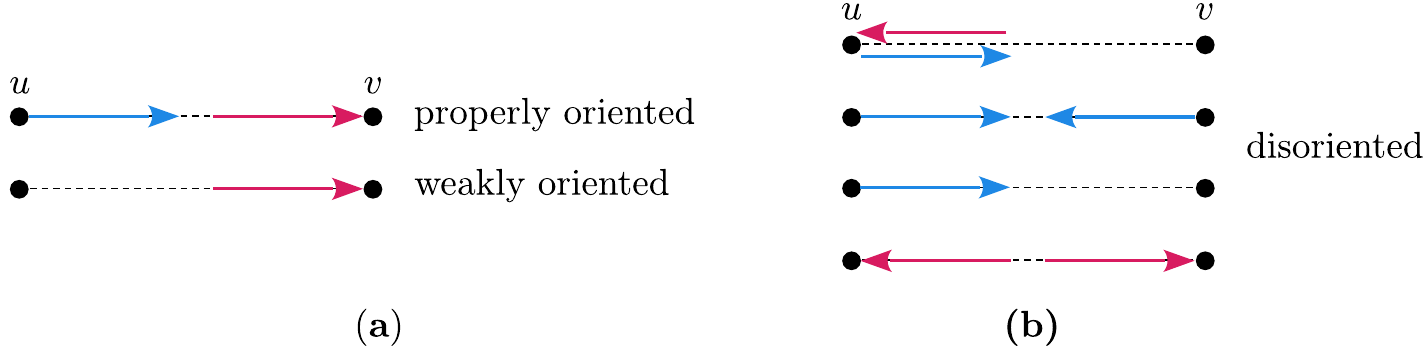}
  \caption{The three {\em orientation statuses} of an edge $\{u,v\}$ induced by the $\children$ and $\parent$ variables of nodes $u$ and $v$. There is an outgoing blue arrow from $u$ if $\parent(u) = \col(v)$, and an incoming red arrow into $v$ if $\col(u) \in \children(v)$.
  (a) Properly and weakly oriented edges. (b) Examples of disoriented edges (i.e., edges that are neither properly nor weakly oriented); here, the list of configurations where the edge $\{u,v\}$ is disoriented is not exhaustive.}
  \label{fig:definition}
\end{figure}

The idea is that if an edge is ``disoriented'', its orientation status is ``ill-defined''; we will show that the number of disoriented edges quickly falls to zero. Once this happens, we have edges that are either weakly oriented or properly oriented.
The algorithm then aims to properly orient all the edges. When this finally happens, we have an oriented tree, where all edges point towards a unique root~node.

Note that when nodes $u$ and $v$ interact, they can easily check if the edge $\{u,v\}$ is properly or weakly oriented in either direction. If the edge is not properly oriented, the \mainalgo (given below) simply updates the variables~$\parent$ and~$\children$ of both~$u$ and~$v$ to properly orient the edge. More precisely, if $\col(v) \in \children(u)$ prior to the interaction, then the edge is properly oriented from~$u$ to~$v$, and otherwise it is properly oriented from~$v$ to~$u$.
The algorithm is described in full detail in the pseudo-code below.

\begin{algorithm}[!ht]
  \DontPrintSemicolon

  \SetKwProg{Fn}{procedure}{}{}
  \SetKwFunction{MyProc}{Set-Edge-Orientation}%
  \Fn{\MyProc{$x,y$}}{
    $\parent(x) \gets \col(y)$ \;
    $\children(x) \gets \children(x) \setminus \{ \col(y)\}$ \tcp*{$x$ removes $y$ from its children}
    $\children(y) \gets \children(y) \cup \{ \col(x)\}$ \tcp*{$y$ adds $x$ to its children}
    \If{$\parent(y) = \col(x)$}{
      $\parent(y) \gets \bot$  \tcp*{only modify $\parent(y)$ if necessary}
    }
  }

  Let~$u$ be the initiator and~$v$ the responder (chosen randomly). \;
  \If{the activated edge $\{u,v\}$ is not properly oriented in either direction}{
    \If{$\col(v) \in \children(u)$}{ %
      $\orientalgo(u,v)$ \;
    }
    \Else{
      $\orientalgo(v,u)$ \;
    }
  }
  \caption{{\bf The \mainalgo}}
  \label{alg:tree-orientation}
\end{algorithm}

The algorithm needs to break symmetry, i.e., it assumes an initiator and a responder.
For this, we use randomness.
In fact, a single random bit suffices, at the cost of only a constant factor increase in the expected running time. This is compatible with the assumptions stated in \Cref{sec:model_definition}.

\subparagraph{Remark.} After the \mainalgo stabilises, it is possible that the root~$u$ has its~$\parent$ variable set to some colour~$c$ even though it has no neighbour of this colour. That is, the root may not be aware that it is the root.
This is necessary, as the impossibility result of Angluin et al.~\cite{angluin_selfstabilising_2008} shows that there is no \emph{self-stabilising} leader election protocol that works for all trees.
Similarly, for any node~$w$, it may be the case that~$c \in \children(w)$ although~$w$ does not have any neighbour~$v$ with~$\col(v) = c$. In particular, this implies that nodes never know their exact set of children. %

\subsection{The analysis of the tree orientation protocol}

\subparagraph{Removal of disoriented edges.}
Let $\mathcal{D}_t$ be the set of disoriented edges at time $t \ge 0$.
We show that the edges cannot remain disoriented for too long: the set $\mathcal{D}_t$ is monotonically non-increasing for~$t \ge 0$.
First we argue that an edge cannot become disoriented again if it has already become weakly or properly oriented. However, it can go from being properly oriented to weakly oriented.

\begin{lemma} \label{lemma:1edge-non-degeneracy}
  If $e \in E$ is not disoriented at time~$t$, then~$e$ is not disoriented at time~$t+1$.
\end{lemma}
\begin{proof}
  Suppose $e = \{u,v\}$ is not disoriented at time $t$. We proceed by case distinction:
  \begin{enumerate}
    \item If~$e = e_t$, i.e., $e$ is sampled to interact at time~$t$, then by construction of the \mainalgo the edge becomes properly oriented in this step, because the procedure \orientalgo orients the edge.
      In particular, $e$ is not disoriented at time $t+1$.

    \item Suppose $e_t \cap e = \{v\}$, that is, an edge adjacent to $e$ is sampled at time $t$.
      Let $e_t = \{v,w\}$. Clearly, the state variables of $u$ remain unchanged at step $t$, so the edge $e = \{u,v\}$ can become disoriented only from changes to the state of $v$. We consider two subcases:

      \begin{enumerate}%
        \item If $e_t$ is properly oriented, then $v$ does not change its state variables by definition of the algorithm. Hence, the orientation status of the edge $\{u,v\}$ does not change.

        \item Suppose $e_t$ is not properly oriented (i.e., either weakly oriented or disoriented).
          By definition of the algorithm and the assumption that $\col$ is a 2-hop colouring of $G$, $\col(u) \neq \col(w)$, so
          node $v$ does not add or remove $\col(u)$ from its $\children_t(v)$ set in this time step when running
          \orientalgo. %
          If $v$ modifies its parent variable at this time step, then it sets $\parent$ to $\col(w) \neq \col(u)$.
          In this case, the edge $\{u,v\}$ is weakly oriented.
      \end{enumerate}

    \item In the remaining case, neither $u$ nor $v$ interacts at time step $t$. Hence, the states of~$u$ and~$v$ are unchanged and~$e$ remains oriented at time $t+1$. \qedhere
  \end{enumerate}
\end{proof}

Now we show that all disoriented edges are removed fast.
Let~$T$ be the first time step when all disoriented edges have been permanently removed, that is,
\begin{equation*}
  T := \inf \{t \geq 0 : \mathcal{D}_{s} = \emptyset \textrm{ for all } s \ge t \}.
\end{equation*}

\begin{lemma} \label{lemma:non-degeneracy}
  The random variable $T$ is $O(n \log n)$ with high probability.
\end{lemma}
\begin{proof}
  When an edge is sampled, it becomes oriented by construction of the \mainalgo.
  By \Cref{lemma:1edge-non-degeneracy}, the edge remains  oriented upon all subsequent interactions.
  Because $G$ is a tree,
  by the coupon collector problem,
  all $n-1$ edges are sampled at least once within $O(n \log n)$ steps with high probability. Thus, by this time, all edges have become oriented and remain oriented~thereafter.
\end{proof}

\subparagraph{Tracking weakness.}
For the sake of the analysis, we place on each edge a distinct \emph{edge-marker} from a set $\M$ of $|E|$ markers. Let $Y_{0} \colon \M \rightarrow E$ be the bijection that gives the (arbitrary) initial location of each edge-marker.
The process $(Y_t)_{t \ge 0}$ is defined inductively as follows:
\begin{enumerate}
  \item If the edge $e_t = \{u,v\}$ sampled to interact at time~$t$ is \emph{weakly} oriented from $u$ to $v$, and $v$ is incident to a \emph{properly} oriented edge $e' = \{v,w\}$ from $v$ to $w$, then let $x = Y_t^{-1}(e_t)$ and $x' = Y_t^{-1}(e')$ be the edge markers on $e_t$ and $e'$ respectively. We define $Y_{t+1}(x) := e'$,  $Y_{t+1}(x') := e_t$, and $Y_{t+1}(x'') := Y_{t}(x'')$ for all $x'' \in \M \setminus \{x,x'\}$. That is, we swap the edge-markers on $e_t$ and $e'$ while keeping all other edge-markers as is. Note that $Y_{t+1}$ remains a bijection by construction.
  \item Otherwise, $Y_{t+1} := Y_t$. That is, no edge-marker changes its location.
\end{enumerate}
From now on, we focus on the execution after time~$T$ when there are no more disoriented edges.
We use the  edge-markers to ``track'' improperly oriented edges in the following sense: For any edge-marker $x \in \M$, if the edge $Y_t(x)$ is properly oriented at time~$t \ge T$, then~$Y_{t+1}(x)$ is properly oriented at time~$t+1$.
Equivalently, this means that if~$Y_t(x)$ is weakly oriented at time~$t$, then~$x$ remains on a weakly oriented edge until the weakness is resolved once and for all.

We show that the weakness must be resolved by the time the marker hits an edge connected to a leaf. Intuitively, the weakness ``leaks out'' of the tree at the leaves.
To complete the analysis, we need to bound the time until this happens using a potential argument.

\subparagraph{The marker potential.}

Consider an edge $e = \{u,v\}$ that is weakly oriented from $u$ to $v$ at time~$t \ge T$.
Since $G = (V,E)$ is a tree, the subgraph $G' = (V, E \setminus \{e\})$ of $G$ consists of two connected components.
We write $D_t(e)$ for the longest distance between~$v$ and a node in the connected component of~$G'$ that contains~$v$:
\begin{equation*}
  D_t(e) := \max_{\substack{w \in V \text{ s.t. } \\ \dis_{G'}(v,w) < +\infty}} \dis_{G'}(v,w).
\end{equation*}
For an edge-marker $x$ on the edge $Y_t(x) = e$ at time $t \ge T$, we define the potential as
\[
  \Phi_t(x) :=
  \begin{cases}
    1+D_t(e) & \textrm{if } e \text{ is not properly oriented at time } t, \\
    0 & \textrm{otherwise.}
  \end{cases}
\]
Intuitively, $D_t(e)$ gives the maximum distance that an edge-marker can travel by moving from~$u$ in the direction of~$v$ in the tree without ``turning back''; see \Cref{figure:definition_potential} for an illustration.
We now show that the potential of each edge-marker $x$ is non-increasing in~$t$ and that the potential decreases by at least one unit whenever the corresponding edge $Y_t(x)$ is weakly oriented and sampled at time $t$.

\begin{figure}[htbp]
  \centering
  \includegraphics[width=0.6\linewidth]{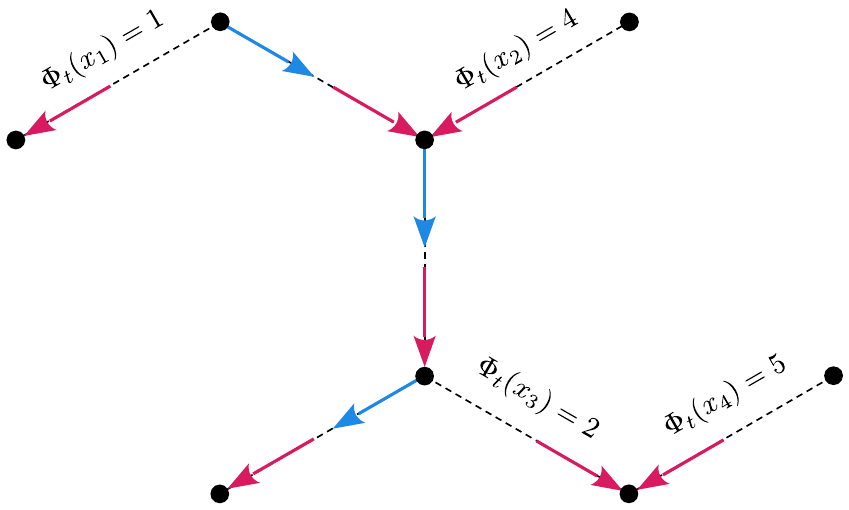}
  \caption{The potential~$\Phi_t(\cdot)$. Here $\mathcal{M} = \{x_1, \ldots, x_7\}$ are the edge-markers. We adopt the same convention as in \Cref{fig:definition} to represent the~$\parent$ and~$\children$ variables of every node involved. This tree has 3 properly oriented edges and 4 weakly oriented edges. The potential of the three markers $x_5,x_6,x_7$ on the properly oriented edges is zero and not drawn on the~figure.}
  \label{figure:definition_potential}
\end{figure}

\begin{lemma} \label{lemma:decreasing_potential}
  Let $t \ge T$. Then for all $x \in \M$ the potential satisfies
  \begin{equation*}
    \Phi_{t+1}(x) \leq
    \begin{cases}
      \max(0,\Phi_t(x)-1) & \text{if~$Y_t(x) = e_t$} \\
      \Phi_t(x) & \text{otherwise.}
    \end{cases}
  \end{equation*}
\end{lemma}
\begin{proof}
  Let $e_t$ be the edge sampled to interact at time~$t$, and $x$ be the edge-marker on $e_t$ at time $t$, i.e., $Y_t(x) = e_t$.
  First, suppose that $e_t$ is properly oriented at time $t$. Then the \mainalgo does not update any state variables, so the edge remains properly oriented at time $t+1$. Because $e_t$ was not weakly oriented, we get  that $Y_{t+1} = Y_t$. Therefore, $\Phi_{t+1}(x) = 0$ by definition of the potential, and the claim trivially follows.

  Now consider the other case that $e_t$ is not properly oriented at time $t$.
  Hence, $\Phi_t(x) \ge 1$. Without loss of generality, suppose that the edge $e_t = \{u,v\}$ is weakly oriented from $u$ to $v$.
  We distinguish two subcases; see \Cref{fig:proof_sketch} for an illustration:
  \begin{enumerate}
    \item Suppose that $v$ has an incident edge $e' = \{v,w\}$ that is properly oriented from $v$ to $w$. Let $x'$ be the edge-marker on $e'$ at time $t$, i.e., $Y_t(x')=e'$. By definition of the process $(Y_t)_{t \ge T}$, the markers $x$ and $x'$ are swapped in $Y_{t+1}$. Moreover, the algorithm updates the value $\parent(v)$ from $\col(w)$ to $\col(u)$. Hence, $\{u,v\}$ becomes properly oriented from $v$ to $u$, whereas $\{v,w\}$ becomes weakly (improperly) oriented from $v$ to $w$ at time $t+1$.

      In particular, the marker $x$ moves from $\{u,v\}$ to $\{v,w\}$, while the orientation of $\{u,v\}$ at time $t$ is the same as the orientation of $\{v,w\}$ at time $t+1$. In particular, this means that $\Phi_{t+1}(x) \leq \Phi_t(x)-1$.
      Meanwhile, $x'$ remains on a properly oriented edge at both time steps~$t$ and~$t+1$, so $\Phi_{t}(x') = \Phi_{t+1}(x') = 0$.

      Finally, note that the algorithm only modifies the~$\parent$ variables of~$u$ and~$v$ as well as their~$\children$ sets using the respective colours of $u$ and $v$. Therefore, the orientation status of all other edges apart from $\{u,v\}$ and $\{v,w\}$ remains unchanged. In particular, this means that the potential of any other edge-marker $x'' \in \M \setminus \{x,x'\}$ remains unchanged, i.e., $\Phi_{t+1}(x'') = \Phi_t(x'')$.

    \item Suppose that $v$ has no incident edge that is properly oriented from $v$ to some neighbour $w$. By definition of the $(Y_t)_{t \ge T}$ process, we have that $Y_{t+1} = Y_t$. Moreover, the edge $\{u,v\}$ becomes properly oriented by time $t+1$, so $\Phi_{t+1}(x) = 0 \le \Phi_t(x) - 1$. (In words, the weakness tracked by the edge-marker $x$ has been resolved and $x$ will remain on an oriented edge for the rest of the execution.)

      By construction of the algorithm, only the orientation status of the edge $\{u,v\}$ changes in time step $t$. Hence, the potential of any other edge-marker $x' \neq x$ remains unchanged, i.e., $\Phi_{t+1}(x') = \Phi_t(x')$.
      \qedhere
  \end{enumerate}

  \begin{figure}
    \centering
    \includegraphics[width=0.99\textwidth]{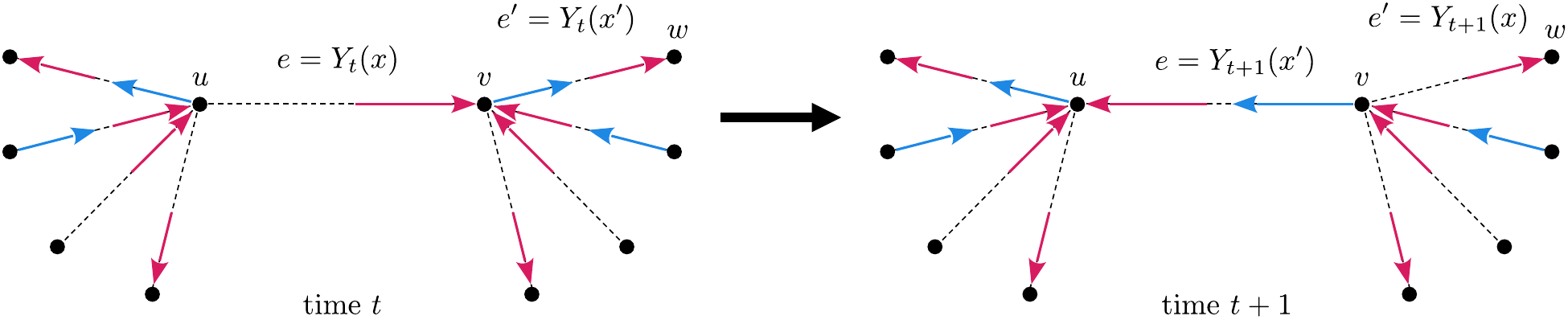}
    \caption{The dynamics of the \mainalgo when a weakly oriented edge $e$ is sampled to interact at time~$t$ after there are no more disoriented edges. We use the same convention as in \Cref{fig:definition} to represent the~$\parent$ and~$\children$ variables of nodes. The nodes $u$ and~$v$ have in addition properly and weakly oriented edges in both possible directions. Assuming the underlying 2-hop colouring is stable, one can check that all edges other than~$e$ and~$e'$ are not modified.}
    \label{fig:proof_sketch}
  \end{figure}
\end{proof}

\subsection{\texorpdfstring{Proof of \Cref{thm:coloured-tree-orientation}}{Proof of Theorem 2}}

As a consequence of \Cref{lemma:decreasing_potential}, the stabilisation time of the \mainalgo is bounded by the time needed to sample each edge-marker~$D$ times, and our upper bound follows from a simple result on a variant of the coupon collector problem.

\begin{proof}
  The state complexity bound follows by observing that storing the $\col$ takes $\chi$ states and storing $\parent$ takes $\chi+1$ states. Storing the $\children$ set takes $2^\chi$ states.

  Recall that~$T$ is the first time step from which all edges remain oriented forever. With high probability, we have $T \in O(n \log n)$ by \Cref{lemma:non-degeneracy}.
  Consider any edge-marker $x$ and let $k = \max\{D, \log n\}$. We now show that after time $T$, the potential of marker $x$ reaches zero in $O(kn)$ steps with high probability.
  Let $T'>T$ be the first time step when $x$ has been sampled at least $k$ times after time $T$.
  Since $\Phi_{T}(x) \le D$, by applying \Cref{lemma:decreasing_potential} on the interval $(T, T']$, we get that
  \begin{equation}
    \Phi_{T'}(x) \leq \max(0,\Phi_{T}(x) - k) = 0.
  \end{equation}
  The random time $S=T'-T$ is stochastically dominated by a sum of geometric random variables $Z = Z_1 + \cdots + Z_k$, where the success probability of $Z_i$ is $1/m$.
  For a constant $\lambda>1$, \Cref{lemma:sum-of-geometric} gives
  \[
    \Pr[S \ge \lambda kn] \le \Pr[Z \ge \lambda \Exp[Z]] \le \exp(-\Omega(k)) \le 1/n^{\Omega(1)}.
  \]
  Thus, with high probability, after $O(kn)$ steps, the marker $x$ has zero potential.
  Taking the union bound over all markers $x$ yields that all markers have zero potential by time $O(kn)$ with high probability.

  Let $T^\star$ be the first time when all markers have zero potential.
  We argue that the \mainalgo has stabilised by time $T^\star$.
  Since all markers have zero potential at time $T^\star$,
  this means that at this time all edges are properly oriented, i.e., there are no weakly oriented edges.
  Since the algorithm cannot modify the state of any node in subsequent steps,
  the orientation of no edge can change from this point onward.
  Since the~$\col$ variables form a 2-hop colouring, every node~$u$ has at most one neighbour~$v$ such that~$\col(v) = \parent(u)$. This implies that every node~$u$ has at most one oriented outgoing edge.
  Hence, the directed tree given by the $\parent$ variables has a unique~root.
  Finally, as the protocol is self-stabilising, the expected time bound follows from \Cref{lemma:whp-to-exp}.
\end{proof}

Together with the 2-hop colouring protocol from \Cref{thm:2-hop-colouring}, we get the following result.
\orientsimp*

\section{Leader election on trees}\label{sec:tree-le}

We now give a leader election algorithm for trees. Recall that in the above orientation algorithm, the root does not know that it is the root. Indeed, otherwise we would contradict the impossibility of self-stabilising leader election on trees~\cite{angluin_selfstabilising_2008}.
We now prove the following.

\leub*

The above further implies that on bounded-degree trees leader election can be solved in optimal $\Theta( Dn )$ time with constant-state protocols.

\subparagraph{The leader election protocol.}
We run the tree orientation algorithm given by \Cref{coro:orientation} in parallel to the following simple token dynamics. Initially, each node starts with a token. When an edge oriented from node $u$ to $v$ is sampled, then the following rules are applied:
\begin{enumerate}%
  \item If both $u$ and $v$ have a token, then the token held by node $u$ is annihilated.
  \item If $u$ holds a token and $v$ does not, then the token moves to node $v$.
\end{enumerate}
Otherwise, the tokens do not change state.
If a node holds a token, it outputs that it is a leader.

\begin{proof}[Proof of \Cref{coro:le}]

  It is easy to check that if the protocol is started with a non-empty set of nodes holding a token, then at least one token always remains even if the edge orientations change over time during the stabilisation of the orientation algorithm. Once the orientation stabilises at time $T_1$, the tree has been oriented towards some root node $r \in V$.

  Let $U \subseteq V$ be the set of nodes that hold a token at time $T_1$.
  The token located at node $v \in V$ at time $T_1$ will either reach the root node $r$ or become annihilated by the time step $T_v + T_1$, where $T_v$ is the time until the edges on the path from $v$ to $r$ have been sampled (as a subsequence) in order by the scheduler after time $T_1$.
  Note that $T_v$ is stochastically dominated by a sum of $D$ geometric random variables with success probability $1/m$.
  Choose $\lambda = K + \log n / D \ge 1$ for a sufficiently large constant $K$. This yields $\lambda \Exp[T_v] \le (DK + \log n)m$.
  Now we can apply first part of
  \Cref{lemma:sum-of-geometric} with $\lambda$ to get
  \[
    \Pr[T_v \ge \lambda \Exp[T_v]]  \le \exp(-Dc(\lambda)) \in \exp(-\Omega(\log n)),
  \]
  because $c(\lambda) \in \Omega(\lambda)$ for large enough $K$.

  The protocol stabilises in time $T = \max \{ T_v : v \in U \} + T_1$, as the first token to reach the root will annihilate all other tokens that will eventually arrive at the root. By the union bound, it follows that
  $\max \{ T_v : v \in U \} \in O(Dn + n \log n)$ with high probability.
  Finally, by \Cref{coro:orientation}, we have $T_1 \in O(Dn + \Delta n \log n)$ with high probability, which gives the claimed w.h.p.\ time bounds. The expected time bound now follows easily from \Cref{lemma:whp-to-exp} by observing that any configuration that starts with a token will always have at least one token (hence, it is closed in the sense of \Cref{lemma:whp-to-exp} in the Appendix).
\end{proof}

While leader election cannot be solved in a self-stabilising manner on trees, we note that the above protocol is robust against the choice of initial configuration in the following sense: if at least one agent starts with a token, then a leader will be elected.

\section{Fast exact majority in directed trees} \label{sec:maj}

In this section, we prove the following result for directed trees.

\begin{theorem} \label{thm:majority_upper_bound}
  There is an exact majority protocol for directed trees that stabilises in $O(n^2 \log n)$ steps in expectation and with high probability, and uses $O(1)$ states.
\end{theorem}

Together with the orientation protocol from \Cref{coro:orientation}, this immediately implies the following result stated in \Cref{sec:contributions}.

\majub*

\subsection{The exact majority algorithm} \label{sec:maj_algorithm}

We set out to prove \Cref{thm:majority_upper_bound}.
We now assume that $G$ is an oriented tree with all edges pointing towards a unique root.
In the protocol, each node $v$ stores state variables $\token(v) \in \{ \atype, \btype, \ctype \}$ and  $\out(v) \in \{\atype, \btype\}$.
Here, $\token(v)$ is interpreted as a token, which traverses the tree.
The token $\atype$ corresponds to input 0 and $\btype$ corresponds to input 1.

The protocol is self-stabilising in the following sense: no matter what initial configuration the adversary chooses, the protocol will reach consensus almost surely. The consensus value will be~$\atype$ if $\token_0(\cdot)$ has a strict majority of $\atype$ at time~$0$, and $\btype$ if there is a strict majority of $\btype$. Otherwise, the consensus value can be arbitrary.

We note that prior work has studied (in complete graphs) a stronger variant of self-stabilising exact majority with \emph{fixed inputs} that cannot be affected by transient faults~\cite{kanaya2025time}. However, it is easy to see that this variant with fixed inputs cannot be solved self-stabilisingly on trees, by the same argument as for leader~election.

\subparagraph{The majority algorithm for directed trees.}
The protocol is as follows. When an edge oriented from $u$ to $v$ is sampled, the nodes $u$ and $v$ perform the following in order:
\begin{enumerate}%
  \item If $\{ \token(u), \token(v)\} = \{ \atype, \btype\}$, then both set their tokens to $\ctype$.
  \item If  $\token(u) \in \{\atype,\btype\}$ and $\token(v) = \ctype$, then $u$ and $v$ swap tokens.
  \item If $\token(v) \neq \ctype$, then~$v$ sets $\out(v) \gets \token(v)$.
  \item Finally, node $u$ sets $\out(u) \gets \out(v)$.
\end{enumerate}

The first two rules implement a ``directed'' version of the annihilation dynamics recently studied by Rybicki et al.~\cite{rybicki2026space}. Rule (1) dictates that $\atype$ and $\btype$ annihilate each other when meeting, and Rule (2)
ensures that $\ctype$-tokens are pushed towards the leaves and other tokens are pushed towards the root.
We can show that the directed version performs much faster than the undirected version on trees.

The last two rules only handle output propagation:
once the protocol stabilises, the root will hold either an $\atype$- or a $\btype$-token corresponding to the majority input. Any other node will either hold a token of the same type as the root or a $\ctype$-token.
The rules (3)--(4) then propagate the token value held by the root towards all the children.

\subsection{The analysis for strict majority}

For the analysis, we consider a set $\Z$ of tokens with~$|\Z|=n$.
Specifically, for a token~$z \in \Z$, we write~$X_t(z) \in V$ to denote the position of~$z$ in the graph~$G$ at time~$t$, and~$M_t(z) \in \{\atype, \btype, \ctype\}$ to denote its type. There is always exactly one token on each node, i.e., $X_t : \Z \to V$ is a bijection.
Whenever an edge~$\{u,v\}$ is sampled to interact, the two tokens located on $u$ and $v$ change type (to~$\ctype$) if Rule~(1) of the algorithm is used, and they exchange their position if Rule~(2) is used.
Let~$r \in V$ be the root of~$G$, and for a token~$z \in \Z$, let
\begin{equation*}
  \depth_t(z) := \dis(X_t(z),r).
\end{equation*}
We assume, without loss of generality, that~$\atype$ is the majority species.
Rule (1) of the algorithm removes exactly one token of type~$\atype$ and one token of type~$\btype$; therefore, the difference between the number of tokens with type $\atype$ and $\btype$ remains the same throughout the execution.
In particular, $\atype$ always remains the majority species.
Importantly, this holds even when the tree is not properly oriented, so the \majalgo can be safely composed with any self-stabilising tree orientation~protocol.

\begin{definition} [Fair schedule]
  We say that a schedule~$(e_t)_{t \geq 0}$ is {\em fair} if it contains every edge of~$G$ infinitely often. Given a fair schedule~$(e_t)_{t \geq 0}$, we partition it into asynchronous {\em rounds} of the form~$\{t_k,\ldots,t_{k+1}-1\}$, where $(t_k)_{k \geq 0}$ is defined inductively as follows: $t_0 = 0$, and for every~$k \geq 0$, $t_{k+1}$ is the first time step when all edges of~$G$ have been sampled at least once after time~$t_{k}$.
\end{definition}

Note that the stochastic schedule is fair with probability 1.
For the rest of the analysis,
we fix the schedule~$(e_t)_{t \geq 0}$ and assume that it is fair.
Formally, fixing the entire schedule is necessary to determine whether the following properties hold at any given time.

\begin{definition} \label{def:blocked_cancelled_active}
  We say that~$z$ is {\em cancelled} at time~$t$ if $M_t(z) = \ctype$. We say that a token~$z \in \Z$ is {\em blocked} at time~$t$ if $M_t(z) \in \{\atype,\btype\}$ and for every time~$s\geq t$, $X_s(z) = X_t(z)$ and~$M_s(z) = M_t(z)$. (In other words, $z$ is blocked if it will never be cancelled and has already reached its final position in the tree.)
  Finally, we say that~$z$ is {\em active} at time~$t$ if it is neither blocked nor cancelled.
\end{definition}

We now state a few simple propositions that will be useful throughout the analysis.
\begin{lemma}[Basic observations] \label{lemma:basic_obs}
  For every~$z,z' \in \Z$ and every~$t > 0$:
  \begin{enumerate}%
    \item If~$z$ is not cancelled at time~$t$ then~$\depth_t(z) \leq \depth_{t-1}(z)$, otherwise~$\depth_t(z) \geq \depth_{t-1}(z)$. \label{item:direction}
    \item If~$z$ and~$z'$ are not cancelled at time~$t$ and there is a directed path from~$X_{t-1}(z)$ to~$X_{t-1}(z')$, then there is a directed path from~$X_t(z)$ to~$X_t(z')$. \label{item:path_persistency}
    \item If~$z$ is blocked at time~$t$, then~$M_s(z) = \atype$ for every~$s \geq 0$. \label{item:blocked_condition}
  \end{enumerate}
\end{lemma}
\begin{proof}
  By construction of the \majalgo, $\ctype$-tokens can only move away from the root, and $\atype$- and $\btype$-tokens can only move towards the root, hence claim~\ref{item:direction} holds.

  Now, consider the case that $z$ and~$z'$ are not cancelled at time~$t$ and there is a directed path from~$X_{t-1}(z)$ to~$X_{t-1}(z')$. Then by claim~\ref{item:direction}, both tokens may only move towards the root, which establishes claim~\ref{item:path_persistency}.

  For claim~\ref{item:blocked_condition}, first note that in the final configuration, tokens of type~$\atype$ and $\btype$ cannot be adjacent to each other; otherwise the corresponding interaction would lead to a different configuration. For the same reason, they cannot have a parent holding a $\ctype$-token. In particular, this implies that the parent of a $\btype$-token must also hold a $\btype$-token, and the child of a $\ctype$-token must also hold a $\ctype$-token.
  Assume, for the sake of contradiction, that there is at least one $\btype$-token in the final configuration. By induction, the observation above implies that all nodes on the path from any such token to the root also hold a $\btype$-token. In other words, nodes with a $\btype$-token form a unique connected component containing the root. The observation above also implies that all other nodes hold $\ctype$-tokens, i.e., there are no more $\atype$-tokens in the tree. Since $\atype$-tokens are initially a majority, and since our algorithm proceeds only in one-to-one cancellations, this leads to a contradiction and concludes the proof of claim~\ref{item:blocked_condition}.
\end{proof}

\begin{definition}
  For~$z \neq z' \in \Z$, we write~$z \succ z'$ if there is a time~$t \geq 0$ such that $z$ and~$z'$ are not cancelled at time~$t$, and there is a directed path from~$X_t(z)$ to~$X_t(z')$.
\end{definition}

\begin{lemma} \label{lemma:acyclic}
  The relation~$\succ$ is acyclic: there is no subset~$\{ z_1,\ldots,z_\ell \} \subseteq \Z$ with~$\ell>1$ such that $z_1 \succ z_2 \succ \ldots \succ z_\ell \succ z_1$. In particular, $\succ$ is asymmetric.
\end{lemma}
\begin{proof}
  Assume, for the sake of contradiction, that~$\succ$ is not acyclic. Consider a subset~$\{ z_1,\ldots,z_\ell \} \subseteq \Z$ of minimal size such that $z_1 \succ z_2 \succ \ldots \succ z_\ell \succ z_1$. In what follows, we treat~$1$ and~$(\ell+1)$ as a single index to simplify the notations, that is, we write $z_1 = z_{\ell+1}$.
  For~$i \in \{1,\ldots,\ell\}$, let
  \begin{align*}
    T_i &:= \inf \{t \geq 0 : \text{there is a directed path from~$X_t(z_i)$ to~$X_t(z_{i+1})$} \}, \\
    U_i &:= \inf \{t \geq 0 : \text{either token } z_i \text{ or } z_{i+1} \text{ is cancelled at time~$t$} \}.
  \end{align*}
  By definition, $z_i \succ z_{i+1}$ implies that $T_i < +\infty$ and $T_i < U_i$ (it might be the case that~$U_i = +\infty$).
  By \Cref{item:path_persistency} in \Cref{lemma:basic_obs}, we can show that for every~$t \in \{T_i,\ldots,U_i-1\}$, there is a directed path from~$X_t(z_i)$ to~$X_t(z_{i+1})$.
  We now show that there can be no overlap between two such consecutive intervals, i.e., for every~$i \in \{1,\ldots,\ell\}$,
  \begin{equation} \label{eq:no_overlap}
    \{T_i,\ldots,U_i-1\} \cap \{T_{i+1},\ldots,U_{i+1}-1\} = \emptyset.
  \end{equation}
  Assume, for the sake of contradiction, that there is~$i \leq \ell$ and~$t \in \{T_i,\ldots,U_i-1\} \cap \{T_{i+1},\ldots,U_{i+1}-1\}$. Then, there is a directed path from~$X_t(z_i)$ to~$X_t(z_{i+1})$, and a directed path from~$X_t(z_{i+1})$ to~$X_t(z_{i+2})$.
  If~$\ell = 2$, then~$z_{i+2} = z_i$ and hence there is a directed cycle in~$G$, which is a contradiction. Otherwise, note that since~$t < \min(U_i,U_{i+1})$, neither~$z_i$ nor~$z_{i+2}$ are cancelled at time~$t$, which implies that~$z_i \succ z_{i+2}$. In that case, $z_{i+1}$ can be removed from the cycle, which violates the minimality of~$\{ z_1,\ldots,z_\ell\}$ and concludes the proof of \Cref{eq:no_overlap}.

  Without loss of generality (up to re-indexing the tokens), we assume that~$U_1 = \sup_i U_i$.
  We show by induction that $U_{i+1} < U_i$ for every~$i \in \{1,\ldots,\ell\}$.
  By the assumption above, $U_2 \leq U_1$, and hence~$U_2 < U_1$ by \Cref{eq:no_overlap}.
  Now, assume that $U_{i+1} < U_i$ for a given~$i<\ell$. By definition, this implies that~$z_{i+1}$ is not cancelled at time~$U_{i+1}$, and therefore $z_{i+2}$ must be cancelled at time~$U_{i+1}$. In turn, this implies that~$U_{i+2} \leq U_{i+1}$, and by \Cref{eq:no_overlap}, that~$U_{i+2} < U_{i+1}$, which concludes the induction.
  Eventually, we obtain that~$U_1 > U_{\ell+1} = U_1$, which is a contradiction.%
\end{proof}

Now that we are equipped with the partial order~$\succ$ on the set of tokens, we informally describe the main idea behind the proof.
Consider the trajectory of some token~$z_0 \in \Z$, initially with type in~$\{\atype,\btype\}$. By construction of the algorithm, and as long as the token is active, either $z_0$ moves towards the root by one unit in every round; or there exists a round~$k$ and a token~$z_1$ of the same type as~$z_0$, where~$z_1$ is located directly ``below''~$z_0$ in round~$k$, thus preventing $z_0$ from making progress. In the latter case, we have that~$z_0 \succ z_1$, and we can consider $z_1$ instead of $z_0$. Since $\succ$ is acyclic, this process will eventually stop, and we can find a token~$z_i$ whose movement down the tree never gets obstructed by a token of the same type. Such a token must move at a constant speed until it stops being active or reaches the root, and we can therefore upper bound its depth in round~$k$ by~$D-k$.

We can repeat the reasoning with the remaining set of tokens. When doing so, we can more or less ignore $z_i$, knowing that if a token is located next to~$z_i$ in round~$k$, then its depth cannot be much larger than the depth of~$z_i$ in that round, e.g., at most~$D-k+2$.
We refer to tokens whose depth can be upper bounded by~$D-k+2\ell$ as ``$\ell$-leading'', which is defined precisely as follows.

\begin{definition}
  We say that a token~$z \in \Z$ is {\em $\ell$-leading} if for every round~$k \geq \ell$,
  \begin{equation} \label{eq:moving_down}
    \depth_{t_k}(z) \leq D - k + 2\ell, \quad \text{or \quad $z$ is not active at time~$t_k$}.
  \end{equation}
\end{definition}

The next lemma, whose proof is deferred to \Cref{sec:proof_of_induction_step}, formalises the informal explanation~above.

\begin{lemma} \label{lemma:leading_induction}
  For every~$\calL \subset \Z$ of size~$\ell$ such that all tokens in~$\calL$ are $(\ell-1)$-leading, there exists~$z \in \Z \setminus \calL$ such that~$z$ is~$\ell$-leading.
\end{lemma}

Using \Cref{lemma:leading_induction}, we obtain an upper bound on the number of rounds needed to stabilise.

\begin{lemma} \label{thm:majority_main}
  After~$4n$ rounds, the algorithm has stabilised and every node outputs $\atype$.
\end{lemma}
\begin{proof}
  Note that if a token is~$(\ell-1)$-leading, then it is also~$\ell$-leading.
  Using \Cref{lemma:leading_induction}, we can show by induction on~$\ell$ that all tokens in $\Z$ are~$n$-leading.
  Therefore, all tokens satisfy \Cref{eq:moving_down} with~$\ell = n$ in round~$3n$, which implies that none of them is active from round~$3n$ onwards; or in other words, the \majalgo has stabilised.

  By \Cref{item:blocked_condition} in \Cref{lemma:basic_obs}, this implies that all tokens have type in~$\{\atype,\ctype\}$ at that time.
  Since we assumed the number of~$\atype$-tokens to be strictly larger initially than the number of $\btype$-tokens, there must be at least one~$\atype$-token remaining.
  Thus, the token on the root must have type~$\atype$ in the final configuration.
  An additional~$D \leq n$ rounds are enough for the \majalgo to propagate the output of the root to every node, which concludes the proof of \Cref{thm:majority_main}.
\end{proof}

In the case of strict initial majority, the proof of \Cref{thm:majority_upper_bound} now follows from \Cref{thm:majority_main} by a standard analysis of the coupon collector problem: since each edge in the tree is sampled with probability~$1/(n-1)$, with high probability, each of the first~$4n$ rounds lasts at most $O(n \log n)$ time steps.

\subsection{\texorpdfstring{Proof of \Cref{lemma:leading_induction}}{Proof of Lemma 29}} \label{sec:proof_of_induction_step}

In this section, we assume that there exists a subset $\calL \subset \Z$ of size~$\ell$ such that all tokens in~$\calL$ are $(\ell-1)$-leading, and we show that we can find a token $z \in \Z \setminus \calL$ such that~$z$ is~$\ell$-leading.
Let~$\Z_\atype$ be the set of all tokens that do not belong to~$\calL$, and are initially of type~$\atype$, that is
\begin{equation*}
  \Z_\atype := \left\{ z \in \Z \setminus \calL :  M_0(z) = \atype \right\}.
\end{equation*}
Let~$\Z_\atype^\star \subseteq \Z_\atype$ be the set of tokens that are minimal for~$\succ$ in~$\Z_\atype$, that is,
\begin{equation*}
  \Z_\atype^\star := \left\{ z \in \Z_\atype : \text{there is no } z' \in \Z_\atype \text{ such that } z \succ z' \right\}.
\end{equation*}
We define~$\Z_\btype$ and~$\Z_\btype^\star$ similarly.
If~$\Z_\atype \cup \Z_\btype = \emptyset$, then all tokens in~$\Z \setminus \calL$ are initially of type~$\ctype$. In this case, they are all inactive from round~$0$, and therefore all~$\ell$-leading. Hence, from now on, we only consider the case that $\Z_\atype \cup \Z_\btype \neq \emptyset$.
Since~$\succ$ is acyclic by \Cref{lemma:acyclic}, $\Z_\atype \neq \emptyset \implies \Z_\atype^\star \neq \emptyset$ and $\Z_\btype \neq \emptyset \implies \Z_\btype^\star \neq \emptyset$.
With the assumption above, this implies that $\Z_\atype^\star \cup \Z_\btype^\star \neq \emptyset$.

We now give a characterisation of the tokens that belong to $\Z_\atype^\star \cup \Z_\btype^\star$, while not being $\ell$-leading.
Informally, for any such token~$z$, there must be a round~$\kappa(z)$ where either (i) $z$ is obstructed by a {\em blocked} token (in the sense of \Cref{def:blocked_cancelled_active}) of the same type, or (ii) $z$ has reached the root of the tree. Case (i) implies that~$z$ and all tokens that belong to the same subtree in round~$\kappa$ will never be able to move further down the tree. Moreover, if $z \in \Z_\btype^\star$, then only case (ii) is possible, as there can be no blocked token of type $\btype$.

\begin{claim} \label{claim:blocked_from_below}
  If $z \in \Z_\atype^\star \cup \Z_\btype^\star$ is not~$\ell$-leading, then there is a round~$\kappa = \kappa(z) \geq \ell$ such that:
  \begin{enumerate}%
    \item \label{item:depth} The token $z$ satisfies \Cref{eq:moving_down} up to round~$\kappa$ (included). Moreover, $z$ is active at time~$t_{\kappa+1}$,~and
      \begin{equation*}
        \depth_{t_{\kappa+1}}(z) = \depth_{t_\kappa}(z) = D-\kappa+2\ell.
      \end{equation*}
    \item \label{item:blocked} For every token~$z'$ such that ~$z' \succ z$,
      \[
        \min_{t \geq 0} \depth_t(z') \geq  \min_{t \geq 0} \depth_t(z) = D-\kappa+2\ell.
      \]
    \item If~$z \in \Z_\atype^\star$, then $|\Z_\btype|>0$. \label{item:cancel}
    \item If~$z \in \Z_\btype^\star$, then~$D-\kappa+2\ell = 0$. That is, $\kappa = D+2\ell$. \label{item:root}
  \end{enumerate}
\end{claim}
\begin{proof}
  Assume that~$z \in \Z_\atype^\star \cup \Z_\btype^\star$ is not~$\ell$-leading.
  We have that
  \begin{equation*}
    \depth_{t_\ell}(z) \leq D \leq D - \ell + 2\ell,
  \end{equation*}
  so \Cref{eq:moving_down} holds trivially for~$z$ in round~$\ell$.
  Let~$\kappa \geq \ell$ be the last round (after round~$\ell$) in which \Cref{eq:moving_down} holds for~$z$.
  Since $z$ does not satisfy \Cref{eq:moving_down} in round~$\kappa+1$, $z$ is still active at time~$t_{\kappa+1}$~and
  \begin{equation*}
    D-(\kappa+1)+2\ell < \depth_{t_{\kappa+1}}(z) \leq \depth_{t_\kappa}(z) \leq D-\kappa+2\ell.
  \end{equation*}
  This implies that~$\depth_{t_{\kappa+1}}(z) = \depth_{t_\kappa}(z) = D-\kappa+2\ell$, which establishes \Cref{item:depth}.

  Now, consider the case that~$D-\kappa+2\ell > 0$. Let~$u = X_{t_\kappa}(z)$, and let~$(u,v) \in E$ be the unique outgoing edge from~$u$.
  By definition, there is a time~$t \in \{t_\kappa,\ldots,t_{\kappa+1}-1\}$ such that~$e_t = (u,v)$. Since~$z$ does not move and remains active during round~$\kappa$, there must be a token~$w$ with~$X_t(w) = v$ and~$M_t(w) = M_t(z)$ (otherwise, $z$ would move downward or become cancelled at time~$t$). By definition, $z \succ w$, hence $w \in \calL$ by the minimality of~$z$.
  We have
  \begin{equation*}
    \depth_t(w) = \depth_t(z)-1 = D-\kappa+2\ell-1 > D-\kappa+2(\ell-1),
  \end{equation*}
  so in order to satisfy \Cref{eq:moving_down}, $w$ cannot be active at time~$t$. Since we already know that~$w$ is not cancelled, this implies that~$w$ is blocked at time~$t$. Now consider the following cases:
  \begin{itemize}
    \item If~$z \in \Z_\btype^\star$, this yields a contradiction, since there can be no blocked tokens of type~$\btype$.
    \item Now, consider the case that~$z \in \Z_\atype^\star$. By the presence of~$w$, $z$ can never move down after time~$t_{\kappa+1}$, and hence, $\min_{t \geq 0} \depth_t(z) = D-\kappa+2\ell$.
      Now, let~$H$ be the subtree rooted at~$u$. By definition, any token~$z'$ located in~$H$ at time~$t_{\kappa+1}$ satisfies~$z' \succ z$. Moreover, it must remain in~$H$ forever after (by the presence of~$w$), so $\min_{t \geq 0} \depth_t(z') \geq D-\kappa+2\ell$.
      Note that for any token~$z'$ located outside of~$H$ at time~$t_{\kappa+1}$, there can never be a path from~$z'$ to~$z$, and hence we do not have~$z' \succ z$.
      Overall, this establishes~\Cref{item:blocked}.

      Since~$z$ is still active at time~$t_{\kappa+1}$ and will never move down afterwards, the only possibility is that~$z$ is cancelled (from above) at a later time. Therefore, there must exist a token~$z''$ with~$M_0(z'')=\btype$, $\depth_{t_{\kappa+1}}(z'') > D-\kappa+2\ell$, and such that~$z''$ is still active at time~$t_{\kappa+1}$. Note that $z''$ does not satisfy \Cref{eq:moving_down} in round~$\kappa+1$, so~$z'' \in \Z_\btype$, which establishes \Cref{item:cancel}.
  \end{itemize}
  The remaining case is that $D-\kappa+2\ell = 0$. \Cref{item:blocked,item:root} hold trivially in this case, and \Cref{item:cancel} follows from the same argument as above.
\end{proof}

In the remainder of the section, we assume for the sake of contradiction that none of the tokens in~$\Z_\atype^\star \cup \Z_\btype^\star$ are~$\ell$-leading. In particular, all of them must satisfy the properties listed in \Cref{claim:blocked_from_below}. In order to obtain a contradiction, we show that $\Z_\btype^\star$ is in fact a singleton.
\begin{claim} \label{claim:singleton}
  There is a unique element~$z^\star \in \Z_\btype^\star$. Moreover,
  \begin{equation} \label{eq:max_kappa}
    \max_{z \in \Z_\atype^\star} \kappa(z) < \kappa(z^\star) = D+2\ell.
  \end{equation}
\end{claim}
\begin{proof}
  Recall that $\Z_\atype^\star \cup \Z_\btype^\star \neq \emptyset$. Moreover, by \Cref{item:cancel} in \Cref{claim:blocked_from_below}, $\Z_\atype^\star \neq \emptyset \implies \Z_\btype^\star \neq \emptyset$. Therefore, in all cases, $\Z_\btype^\star$ is non-empty.

  \Cref{item:depth,item:root} in \Cref{claim:blocked_from_below} implies that for every~$z \in \Z_\btype^\star$, $\kappa(z) = D+2\ell$ and $\depth_{t_{\kappa(z)}}(z) = 0$.
  Since there can be only one token at depth~$0$ at time~$t_{D+2\ell}$, this implies that~$|\Z_\btype^\star| = 1$.
  Let $z^\star$ be the unique element of $\Z_\btype^\star$; \Cref{eq:max_kappa} also follows from the observation that there can be only one token at depth~$0$ at time~$t_{D+2\ell}$ (and the fact that $\kappa(z)$ can never be strictly greater than $D+2\ell$, which would lead to a contradiction with \Cref{claim:blocked_from_below}).
\end{proof}

We are ready to conclude the proof of \Cref{lemma:leading_induction}.
Since~$z^\star$ is active at the beginning of round~$\kappa(z^\star)+1$ and has depth~$0$, $z^\star$ must be cancelled from above at a later time by some token~$z'$. However, observe that:
\begin{itemize}
  \item If~$z' \in \Z \setminus \calL$, then by \Cref{item:blocked} in \Cref{claim:blocked_from_below}, and by \Cref{claim:singleton}:
    \begin{equation*}
      \min_{t \geq 0} \depth_t(z') \geq D - \max_{z \in \Z_\atype^\star} \kappa(z) +2\ell > 0,
    \end{equation*}
    which yields a contradiction.
  \item If~$z' \in \calL$, let~$k$ be the round in which the cancellation happens.
    We know that~$k \geq \kappa(z^\star)+1$, since~$z^\star$ is still active at the beginning of round~$\kappa(z^\star)+1$.
    At that time, $z'$ is also active. Since~$z'$ satisfies \Cref{eq:moving_down}, we have
    \begin{equation*}
      \depth_{t_{\kappa(z^\star)+1}}(z') \leq D-(\kappa(z^\star)+1)+2\ell = -1,
    \end{equation*}
    which is a contradiction.
\end{itemize}
Thus, $\Z_\atype^\star \cup \Z_\btype^\star$ must contain at least one~$\ell$-leading token, which concludes the proof of \Cref{lemma:leading_induction}.

\subsection{Stabilisation in case of ties}\label{ssec:ties}

Above we assumed that there is a strict majority of $\atype$ tokens.
We show that the process reaches consensus even when there is a tie between $\atype$- and $\btype$-tokens.
This follows from a simple stochastic domination argument.
It suffices to consider the situation where there is at least one $\atype$- and $\btype$-token; otherwise, the bound follows from the fact that the output of the root is propagated to all nodes in $O(D)$ rounds.

Consider an oriented tree and an initial configuration $x_0$ with an equal number of $\atype$- and $\btype$-tokens.
Choose an arbitrary node with a $\btype$-token and change its token to $\ctype$ to obtain a configuration $x'_0$. Now there is a strict $\atype$-majority in $x'_0$.
Consider now the obvious coupling of the two executions $(x_t)_{t \ge 0}$ and $(x'_t)_{t \ge 0}$ using the same schedule.
Define $\sigma_t(v) = (\token_t(v),  \token_t'(v) )$ and
\[
  Z_t(v) =
  \begin{cases}
    0 & \text{if } \sigma_t(v) \in \{ (\atype, \atype), (\btype, \btype), (\ctype, \ctype)\} \\
    1 & \text{otherwise}.
  \end{cases}
\]
Here, $(\sigma_t)_{t \ge 0}$ gives the coupled process and $Z_t(\cdot)$ tracks the nodes where the two processes mismatch (in their token types).
We show that they can mismatch in at most one node at any point during the execution.

\begin{lemma}\label{lemma:coupling}
  For any $t \ge 0$, there exists exactly one $v_t \in V$ such that $Z_t(v_t) \neq 0$.
\end{lemma}
\begin{proof}
  We prove the claim by induction on $t$.
  The base case $t=0$ follows by construction of $x_0$ and $x_0'$.
  For the inductive step, note that if $v_t$ does not interact at time step $t+1$, then the claim follows with $v_{t+1} = v_t$.

  Hence, suppose that $v_t = v$ interacts at time step $t+1$ with node $u$. Note that by the induction hypothesis $Z_t(u) = 0$, that is, $\sigma_t(u)$ consists of a pair of tokens of the same type.
  The values of  $\sigma_{t+1}(v)$ and $\sigma_{t+1}(u)$ are given by the following table, where rows represent possible values of $\sigma_t(v)$ and columns the possible values of $\sigma_t(u)$:
  \[
    \begin{array}{c|cc|cc|cc}
      \sigma_t(v)\backslash \sigma_t(u)
      & \multicolumn{2}{c|}{(\atype,\atype)}
      & \multicolumn{2}{c|}{(\btype,\btype)}
      & \multicolumn{2}{c}{(\ctype,\ctype)}\\
      & u & v & u & v & u & v\\
      \hline
      (\btype,\ctype)
      & (\ctype,\ctype) & (\ctype,\atype) \star
      & (\btype,\ctype) \star & (\btype,\btype)
      & (\ctype,\ctype) & (\btype,\ctype) \star
      \\
      (\ctype,\atype)
      & (\ctype,\atype) \star & (\atype,\atype)
      & (\ctype,\ctype) & (\btype,\ctype) \star
      & (\ctype,\ctype) & (\ctype,\atype) \star
    \end{array}
  \]
  The cells marked by $\star$ denote the node where the discrepancy will be at the end of interaction $t+1$. It is easy to confirm by inspecting the table that only one of the nodes $u$ and $v$ will have the discrepancy at the end of the time step $t+1$.
  Hence, we have that $Z_{t+1}(w) \neq 0$ for exactly one $w \in e_{t+1}$, and $Z_{t+1}(w') =Z_t(w')=0$ for all $w' \notin e_{t+1}$ by the induction hypothesis.
\end{proof}

\begin{lemma}
  From any initial configuration with an equal number of $\atype$- and $\btype$-tokens,
  the algorithm has stabilised after $O(n)$ rounds.
\end{lemma}
\begin{proof}
  By \Cref{thm:majority_main}, we know that after $T \in O(n)$ rounds there are no more $\btype$-tokens in the process started from $x_0'$.
  By \Cref{lemma:coupling}, the process started at $x_0$ has at most one $\btype$-token after $T$ rounds.
  If there is one, then there is also one $\atype$-token. These tokens meet at the root in $O(D)$ rounds, as both can advance towards the root without being blocked.

  Hence, by $O(D + n)$ rounds there are only $\ctype$-tokens in the process started from $x_0$.
  An additional~$D \leq n$ rounds are enough for the algorithm to propagate the output of the root to every node. This gives a total bound of $O(n)$ rounds for the protocol to stabilise.
\end{proof}

Similar to the case of strict majority, the bound of $O(n)$ rounds implies a stabilisation time of $O(n^2 \log n)$ steps with high probability and in expectation in the case of ties as well.

\section{Time lower bounds for trees}\label{sec:lower}

In this section, we prove two time lower bounds for population protocols on trees.

\orientlb*

\lelb*

Technically, the lower bound for leader election given by Alistarh et al.~\cite{alistarh2025near} does not immediately work on trees, because their setup assumes that (a) \emph{all} the nodes of the input graph can be partitioned into sets with isomorphic radius-$\ell$ neighbourhoods and that (b) there are at least two such sets with disjoint radius-$\ell$ neighbourhoods.
Indeed, as leader election can be solved in $O(1)$ time in stars, even an $\Omega(n \log n)$ lower bound cannot in general be true for trees.
We refine these arguments to work in a larger class of graphs by proving a new, more general technical lower-bound lemma.

\subsection{A lower-bound lemma}\label{sec:lower-bound-lemma}

Before we state the lemma, we need to introduce some definitions and notation.

\subparagraph{Sets of influencers.}
We say that a node~$u$ {\em influences} a node~$v$ at time~$t$ if there exists a path~$(u = w_0,w_1,\ldots,w_k = v)$ from~$u$ to~$v$ and an increasing sequence of time steps~$s_0 < \cdots < s_{k-1} < t$ such that $e_{s_i} = \{ w_i,w_{i+1} \}$ for every~$i < k$.
For a set $U \subseteq V$, we write $I_t(U)$ to denote the set of nodes that influence at least one node of~$U$ at time~$t$. The set of influencers $I_t(v)$ gives all the nodes who can in principle influence the state of node $v$ at time $t$; this idea was introduced by Sudo~et~al.~\cite{sudo_leader_2020}. %

\subparagraph{Restrictions and conditional independence.}
For the graph $G=(V,E)$ and any node $v \in V$, we write $N_k(v) = \{ u \in V : \dist(u,v) \le k\}$ for the radius-$k$ neighbourhood of $v$. For a set $U \subseteq V$, we write $N_k(U) = \bigcup_{v \in U} N_k(v)$.
For $f \colon A \to B$ and a set $C \subseteq A$, we write $f_{\restriction C}$ for $f$ restricted to domain $C$.
For an event $\mathcal{E}$ with $\Pr(\mathcal{E})>0$ and two discrete random variables $X$ and~$Y$, recall that~$X$ and~$Y$ are {\em conditionally independent given $\mathcal{E}$} if for every~$x,y$ in their support, the variables satisfy
\begin{equation*}
  \Pr(X = x, Y=y \mid \mathcal{E}) = \Pr(X = x \mid \mathcal{E}) \cdot \Pr(Y = y \mid \mathcal{E}).
\end{equation*}
We are now ready to give the main technical
lemma for deriving our lower bounds.

\begin{lemma}%
  \label{lemma:lower-bound-lemma}
  Let $G = (V,E)$, $V_1, V_2 \subseteq V$ and $k>1$ such that $N_k(V_1) \cap N_k(V_2) = \emptyset$.
  Let $T \sim \mathrm{Poisson}(\lambda)$ for $\lambda > 0$.
  Let $X_0 \colon V \to \Lambda$ be the initial configuration.
  For~$i \in \{1,2\}$, let
  \begin{itemize}
    \item $\mathcal{E}_i$ be the event that $I_T(V_i) \subseteq N_{k-1}(V_i)$,
    \item $X_{i,t} \colon N_{k}(V_i) \to \Lambda$ such that $X_{i,t}(v) = X_t(v)$, and
    \item $Y_{i,t} \colon V_i \to \Lambda$ such that $Y_{i,t}(v) = X_t(v)$.
  \end{itemize}
  Suppose that $X_{1,0}$ and $X_{2,0}$ are independent.
  Then:
  \begin{enumerate}
    \item The variables $Y_{1,T}$ and $Y_{2,T}$ are independent conditioned on the event $\mathcal{E}_1 \cap \mathcal{E}_2$.

    \item If $\phi$ is an isomorphism from $G[N_k(V_1)]$ to $G[N_k(V_2)]$ with $\phi(V_1)=V_2$, and
      $X_{1,0} $ and $X_{2,0} \circ \phi$ are identically distributed,
      then the variables $Y_{1,T}$ and $Y_{2,T} \circ \phi$ are i.i.d.\ conditioned on the event $\mathcal{E}_1 \cap \mathcal{E}_2$.
  \end{enumerate}
\end{lemma}
\begin{proof}
  Let $E_i \subseteq E$ be the set of edges in $G[N_k(V_i)]$ for $i \in \{1,2\}$. Define $E_0 = E \setminus( E_1 \cup E_2 )$.
  Let $p_i = |E_i|/m$ and $T_i \sim \operatorname{Poisson}(p_i \lambda)$ be independent Poisson random variables for $i \in \{0,1,2\}$.
  Since $E_0, E_1, E_2$ are a partition of $E$, we have that $p_0+p_1+p_2 = 1$ and %
  \[
    T = T_0 + T_1 + T_2,
  \]
  where $T \sim \operatorname{Poisson}(\lambda)$.
  Let $\sigma = (e_t)_{1 \le t \le T}$ be the schedule up to time $T$ and let $\sigma_i$ be the subsequence of $\sigma$ consisting of the  $T_i$ edges sampled from $E_i$.
  Since each edge in $\sigma$ is sampled i.i.d.\ from $E$ and $T_1$ and $T_2$ are independent random variables, we get that
  $\sigma_1$ and $\sigma_2$ are independent.
  (We could equivalently first independently sample each $\sigma_i$ and then randomly interleave $\sigma_0, \ldots, \sigma_2$ to obtain $\sigma$.)

  By definition, for $i \in \{1,2\}$, the event $\mathcal{E}_i$ defined as $I_T(V_i) \subseteq N_{k-1}(V_i)$  happens if and only if $\sigma$ does not contain as a subsequence edges on any path from $N_k(V_i) \setminus N_{k-1}(V_i)$ to $V_i$. Clearly, all such paths use only edges from $E_i$. Thus, for $i \in \{1,2\}$, the event $\mathcal{E}_i$ depends deterministically only on $\sigma_i$. Because $\sigma_1$ and $\sigma_2$ are independent, so are the events $\mathcal{E}_1$ and~$\mathcal{E}_2$.

  We assume without loss of generality that all the local randomness used by node $u$ throughout the execution is contained in a single uniform random variable $R(u) \in [0,1)$, where all $(R(u))_{u \in V}$ are i.i.d.\ and sampled independently of the schedule $(e_t)_{t \ge 1}$.
  Now the randomness in the execution of the protocol is entirely contained in $\sigma$, $R$, and $X_0$.
  For~$i \in \{1,2\}$, let $R_i = R_{\restriction N_k(V_i)}$ be the randomness used by nodes in $N_k(V_i)$.
  If $\mathcal{E}_i$ happens, then it is enough to  know the realisation of~$\sigma_i$, $R_i$ and~$X_{i,0}$ to infer the state of any node in~$V_i$ at time~$T$.
  In other words, we can find deterministic functions~$f_i,g_i$ such that
  \begin{equation} \label{eq:conditional_expression_of_X}
    Y_{i,T} =
    \begin{cases}
      f_i(\sigma_i,R_i,X_{i,0}) & \text{if $\mathcal{E}_i$ happens,} \\
      g_i(\sigma,R,X_0) & \text{otherwise.}
    \end{cases}
  \end{equation}
  Let $Z_i = (\sigma_i, R_i, X_{i,0})$.
  The joint distribution of $(Y_{1,T}, Y_{2,T})$ satisfies for all configurations~$x_1,x_2$ on~$V_1$ and~$V_2$ that
  \begin{align*}
    \Pr[Y_{1,T} = x_1, Y_{2,T} = x_2 \mid \mathcal{E}_1 \cap \mathcal{E}_2] &= \frac{\Pr[f_1(Z_1) = x_1, f_2(Z_2) = x_2, \mathcal{E}_1 \cap \mathcal{E}_2]}{\Pr[\mathcal{E}_1 \cap \mathcal{E}_2]} \\
    &= \frac{\Pr[f_1(Z_1) = x_1, \mathcal{E}_1]}{\Pr[\mathcal{E}_1]} \cdot \frac{\Pr[f_2(Z_2) = x_2, \mathcal{E}_2]}{\Pr[\mathcal{E}_2]} \\
    &= \Pr[f_1(Z_1) = x_1 \mid \mathcal{E}_1] \cdot \Pr[f_2(Z_2) = x_2 \mid \mathcal{E}_2] \\
    &= \Pr[Y_{1,T} = x_1 \mid \mathcal{E}_1] \cdot \Pr[Y_{2,T} = x_2 \mid \mathcal{E}_2],
  \end{align*}
  where for the second equality, we have used the independence between $Z_1, \mathcal{E}_1$ and $Z_2,\mathcal{E}_2$; this establishes the first claim of the lemma.

  For the second claim, suppose that $\phi$ is an isomorphism from $G[N_k(V_1)]$ to $G[N_k(V_2)]$ with $\phi(V_1) = V_2$
  and the initial configurations $X_{1,0}$ and $X_{2,0} \circ \phi$ are independent and identically distributed.
  Since $R_1$ and $R_2 \circ \phi$ are also i.i.d., \Cref{eq:conditional_expression_of_X} implies that $Y_{1,T}$ and $Y_{2,T} \circ \phi$ are identically distributed.%
\end{proof}

Recall that if the initial configuration is a deterministic constant function, then $X_{1,0}$ and $X_{2,0}$ in the above lemma are independent (in the formal sense).

\subsection{Lower bound construction}\label{sec:lowerbounds}

In this section, we describe a specific graph on which to apply \Cref{lemma:lower-bound-lemma}. We then establish tight lower bounds for stable tree orientation and leader election.
Given~$n,k \in \bbN$ such that~$\log n \leq k \leq n/8$, let~$T_{n,k}$ be the tree obtained by taking a path~$(v_0,\ldots,v_{8k-1})$ of length~$8k-1$, and attaching two balanced binary trees with~$(n-8k+1)/2$ nodes (not necessarily complete) on each side, i.e., to~$v_0$ and~$v_{8k-1}$ respectively.
The depth of each binary tree is bounded by~$\log((n-8k)/2) \leq \log n \leq k$, so~$\diam(T_{n,k}) \in \Theta(k)$, and~$T_{n,k}$ has~$n$ nodes by construction.

We partition the nodes of~$T_{n,k}$ as follows: for each~$i \in \{0,\ldots,7\}$, we denote by~$V_i = \{v_{ki},\ldots,v_{k(i+1)-1} \}$ a distinct section of the path, and let $U_1$ (resp.~$U_2$) be the set of nodes that belong to the leftmost (resp. rightmost) binary tree, as illustrated in \Cref{fig:illustration_Tnk}.

\begin{figure}[htbp]
  \centering
  \includegraphics[width=0.5\linewidth]{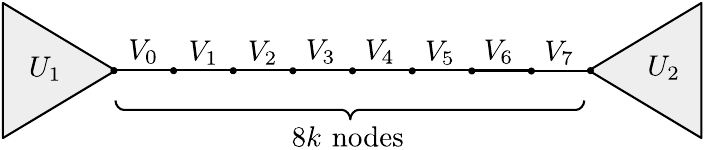}
  \caption{The tree~$T_{n,k}$ and the node partition.}
  \label{fig:illustration_Tnk}
\end{figure}

\begin{lemma} \label{lemma:zooming_in_Tnk}
  Let~$\theta = kn/2$ and $T \sim \mathrm{Poisson}(\theta/2)$. For any subset~$A \in \{V_0,\ldots,V_7\} \cup \{U_1, U_2\}$,
  \begin{equation*}
    \Pr[I_{T}(A) \subseteq N_{k-1}(A)] \geq 1-o(1).
  \end{equation*}
\end{lemma}
\begin{proof}
  We first prove the claim for~$V_1$. Note that~$N_k(V_1) = \{v_0,\ldots,v_{3k-1}\}$.
  The time~$\tau$ until the schedule contains a subsequence of edges forming a path from~$v_0$ to~$v_k$ is bounded by a sum of~$k$ geometric random variables with parameter~$1/m = 1/(n-1)$.
  Hence, by \Cref{lemma:sum-of-geometric}(b), we get
  \begin{equation*}
    \Pr[\tau \leq \theta] = \Pr[\tau \leq kn/2] \leq \exp\pa{-k \cdot (-\tfrac{1}{2} + \ln 2)} \in o(1).
  \end{equation*}
  The same holds for the path from~$v_{3k-1}$ to~$v_{2k-1}$. Therefore, by the union bound, we get
  \begin{equation*}
    \Pr \left[ I_\theta(V_1) \subseteq N_{k-1}(V_1) \right] \geq 1-o(1).
  \end{equation*}
  Moreover, $T \leq \theta$ with high probability (this follows immediately from a classical tail bound for Poisson random variables; see, for example, \cite[Theorem 5.4]{mitzenmacher_probability_2017}). We finally conclude the proof by using another union bound. The exact same arguments can be made symmetrically for every~$A \in \{V_2,\ldots,V_6\}$. For~$U_1$ (resp.~$V_0$), it is enough that the schedule does not contain a subsequence of edges forming a path from~$v_{k-1}$ to~$v_0$ (resp.~$v_{2k-1}$ to~$v_{k-1}$); and hence the same bound holds. Finally, the case of~$U_2$ (resp.~$V_7$) is symmetric to the case of~$U_1$ (resp.~$V_0$).
\end{proof}

\Cref{thm:orientation-lb} is now a direct consequence of the following result.
\begin{theorem} \label{thm:lower_bound_tree_orientation}
  Assume that the initial states of all nodes are i.i.d.
  Any stabilising tree orientation algorithm on $T_{n,k}$ requires $\Omega(kn)$ expected steps.
\end{theorem}
\begin{proof}
  By construction, $N_k(V_1)$ and $N_k(V_4)$ are disjoint.
  Let~$\phi,\phi' : N_k(V_1) \to N_k(V_4)$, where for every~$i \in \{0,\ldots,3k-1\}$, we define
  \begin{align*}
    \phi(v_i) &= v_{3k+i}, \\
    \phi'(v_i) &= v_{6k-i-1}.
  \end{align*}
  The restrictions of~$T_{n,k}$ to both $N_k(V_1)$ and $N_k(V_4)$ are paths of length~$3k-1$, and $\phi$, $\phi'$ are isomorphisms by construction.
  Let~$\theta = kn/2$, let $T \sim \mathrm{Poisson}(\theta/2)$ and let $\mathcal{E}$ be the event
  \begin{equation*}
    \mathcal{E} := \bigcap_{i\in \{1,4\}} \{ I_T(V_i) \subseteq N_{k-1}(V_i) \}.
  \end{equation*}
  By \Cref{lemma:zooming_in_Tnk} and the union bound, $\Pr[\mathcal{E}] \geq 1-o(1)$. From now on, we condition on this event.
  Given that the initial states of all nodes are i.i.d., we can apply \Cref{lemma:lower-bound-lemma} with either~$\phi$ or~$\phi'$: the configurations on~$V_1$ and~$V_4$ at time~$T$ are i.i.d. For $i \in \{1,4\}$, let $A_i$ be the event that $V_i$ is properly oriented and contains a root at time $T$, and let $p = \Pr[A_i \mid \mathcal{E}]$. By independence, we have
  \begin{equation} \label{eq:problem_1}
    \Pr[A_1 \cap A_4 \mid \mathcal{E}] = p^2.
  \end{equation}
  For $i\in \{1,4\}$, let~$R_i$ be the event that $V_i$ is consistently oriented from left to right (that is, each~$\{v_j,v_{j+1}\}$ is oriented from~$v_j$ to~$v_{j+1}$). Similarly, let~$L_i$ be the event that $V_i$ is consistently oriented from right to left.
  When applying \Cref{lemma:lower-bound-lemma} with~$\phi$, we obtain
  \begin{align*}
    \Pr[L_1 \mid \mathcal{E}] &= \Pr[L_4 \mid \mathcal{E}], \\
    \Pr[R_1 \mid \mathcal{E}] &= \Pr[R_4 \mid \mathcal{E}],
  \end{align*}
  and when applying \Cref{lemma:lower-bound-lemma} with~$\phi'$, we get
  \begin{equation*}
    \Pr[L_1 \mid \mathcal{E}] = \Pr[R_4 \mid \mathcal{E}].
  \end{equation*}
  In other words, conditioning on~$\mathcal{E}$, we have that~$L_1$, $R_1$, $L_4$ and $R_4$ all happen with the same probability~$q$.
  Moreover, by independence, it follows that
  \begin{equation} \label{eq:problem_2}
    \Pr[L_1 \cap R_4 \mid \mathcal{E}] = \Pr[L_1 \mid \mathcal{E}] \cdot \Pr[R_4 \mid \mathcal{E}] = q^2.
  \end{equation}
  If $p+2q \leq 3/4$, then $\Pr[\overline{A_1} \cap \overline{R_1} \cap \overline{L_1} \mid \mathcal{E}] \geq 1/4$.
  If this happens, $V_1$ does not contain a root and is not consistently oriented in either direction, and hence the protocol has not stabilised at time~$T$.
  Conversely, if $p+2q > 3/4$, then $\max(p,q) > 1/4$ and by \Cref{eq:problem_1,eq:problem_2}, this implies that either~$A_1 \cap A_4$ (if $p > 1/4$) or $L_1 \cap R_4$ (if $q > 1/4$) happens with probability at least $1/16$ (neither of which is compatible with stabilisation).
  \Cref{thm:lower_bound_tree_orientation} follows by observing that~$\Exp[T] \in \Theta(kn)$.
\end{proof}

Similarly, \Cref{thm:le-lb} is a direct consequence of the following result.
\begin{theorem} \label{thm:lower_bound_leader_election}
  Any stable leader election algorithm on $T_{n,k}$ requires $\Omega(kn)$ expected steps to stabilise.
\end{theorem}
\begin{proof}
  We assume that the input is the deterministic constant function $f(v)=1$.
  This implies that the initial states of all nodes are (formally) independent.
  Let~$\theta = kn/2$, let $T \sim \mathrm{Poisson}(\theta/2)$.
  Let~$A^\star$ be the subset of the partition most likely to contain at least one leader at time~$T$, that is,
  \begin{equation*}
    A^\star := \underset{A \in \{V_0,\ldots,V_7\} \cup \{U_1,U_2\}}{\mathrm{argmax}} \Pr[ A \text{ has a leader at time } T ],
  \end{equation*}
  and let~$p^\star$ be the corresponding probability.
  If~$p^\star < 1/11$, then with probability at least~$1-10p^\star > 1/11$, the graph contains no leader at time~$T$, in which case the algorithm has not yet stabilised. From now on, we only consider the case that~$p^\star \geq 1/11$.

  We pair the subsets of the partition as follows: $(U_1,  U_2)$, $(V_0, V_7)$, and for~$i \in \{1,2,3\}$, $(V_i, V_{3+i})$. It is easy to check that for any pair~$(A, B)$, the neighbourhoods $N_k(A)$ and $N_k(B)$ are disjoint and isomorphic. Let~$B^\star$ be such that $(A^\star, B^\star)$ is one of these pairs, and let
  \begin{equation*}
    \mathcal{E} := \{ I_T(A^\star) \subseteq N_{k-1}(A^\star) \} \cap \{ I_T(B^\star) \subseteq N_{k-1}(B^\star) \}.
  \end{equation*}
  By \Cref{lemma:zooming_in_Tnk} and a union bound, $\Pr[\mathcal{E}] \geq 1-o(1)$.
  Given that the initial states of all nodes are i.i.d., by \Cref{lemma:lower-bound-lemma}, the configurations on~$A^\star$ and~$B^\star$ at time~$T$ are i.i.d.\ conditioned on~$\mathcal{E}$. In particular, we get that
  \begin{align*}
    \Pr[A^\star \text{ has a leader at time } T \mid \mathcal{E} ] &= \Pr[B^\star \text{ has a leader at time } T \mid \mathcal{E} ] \\
    &\geq p^\star - o(1) \geq \frac{1}{11} - o(1),
  \end{align*}
  and by independence, it follows that
  \begin{equation*}
    \Pr[A^\star \text{ and } B^\star \text{ both have a leader at time } T \mid \mathcal{E} ]\geq \frac{1}{11^2} - o(1).
  \end{equation*}
  As a consequence, the algorithm has failed to stabilise at time~$T$ with constant probability.
  \Cref{thm:lower_bound_leader_election} follows by observing that~$\Exp[T] \in \Theta(kn)$.
\end{proof}
\section{Conclusions}

In this paper, we developed population protocols that can be combined to solve several fundamental problems on trees: tree orientation, leader election, and exact majority. These protocols are specifically designed for bounded-degree trees, on which they achieve optimal stabilisation time while requiring only a constant number of states. However, they still work on any tree, and the space and time complexity then grow with the maximum degree $\Delta$. %

As a key building block, we introduced a new population protocol for 2-hop colouring on \emph{general} graphs.
We showed that  on bounded-degree graphs the protocol outperforms the state-of-the-art~\cite{Sudo2018_Loosely,Kanaya2024_Almost}, both in terms of time and space complexity.
Indeed, on bounded-degree graphs, the colouring protocol is asymptotically space- and time-optimal.

All of our algorithms are self-stabilising, with the exception of the leader election protocol, as self-stabilising leader election is known to be impossible on trees.
The exact majority protocol is also self-stabilising in the sense that it always reaches consensus from any configuration, and outputs the correct majority if there exists a strict majority of $\atype$ or $\btype$ tokens in the initial configuration.
Moreover, our protocols require knowledge only of~$\Delta$; in particular, they are uniform on bounded-degree graphs and trees.

A major consequence of our results is that there is no space-time trade-off for leader election in the class of bounded-degree trees~\cite{angluin_selfstabilising_2008}. This distinguishes bounded-degree trees from complete graphs, in which both leader election and exact majority exhibit space-time trade-offs in terms of the size $n$ of the population~\cite{doty_stable_2018,alistarh2017time,alistarh2018space,sudo2019logarithmic,berenbrink_optimal_2020,berenbrink2021time,doty2022time}. Moreover, our work shows that on paths, constant-state exact majority protocols can achieve near-optimal time complexity of $O(n^2 \log n)$ steps.

\subparagraph{Open problems.}
We conclude by raising a few interesting open problems for future work:
\begin{enumerate}
  \item Are there constant-state protocols that are time-optimal in \emph{all} trees?
  \item Are there time-optimal tree orientation protocols that use fewer states for large $\Delta$?
  \item Is there a constant-state protocol for exact majority that runs in $O(D n \log n)$ time on trees?
\end{enumerate}

\subparagraph{AI Disclosure: } We used Codex (GPT 5.4) and Claude (Opus 4.8) to assist with proof reading the manuscript and suggesting proof strategies for some lemmas. The table in the proof of Lemma 33 was generated using GPT 5.4. All results were written by the authors. The authors verified the correctness and originality of all content including references.

\bibliography{references}

\appendix
\newpage

\section{Appendix: Probabilistic tools}\label{apx:tools}

We will use the following simple version of Chernoff's bound, in multiplicative form.
\begin{theorem} [Chernoff's bound] \label{thm:chernoff_bound}
  Let~$X^{(1)},\ldots,X^{(n)}$ be independent binary random variables, let~$X = \sum_{i=1}^n X^{(i)}$ and $\mu = \Exp[X]$. Then for every~$\varepsilon \in (0,1)$,
  \begin{equation*}
    \Pr[X < (1-\varepsilon) \mu] \leq \exp \pa{-\varepsilon^2 \cdot \frac{\mu}{2}}.
  \end{equation*}
\end{theorem}

Janson~\cite[Theorem 2.1 and Theorem 3.1]{janson2018tail} gave the following tail bounds for sums of geometric random variables.

\begin{theorem}\label{lemma:sum-of-geometric}
  Let $p_1, \ldots, p_k \in (0,1]$ and $X = Y_1 + \ldots + Y_k$ be a sum of independent geometric random variables with $Y_i \sim \Geom(p_i)$. Define $p = \min \{ p_i : 1 \le i \le k \}$ and $c(\lambda) = \lambda - 1 - \ln \lambda$. Then
  \begin{enumerate}%
    \item $\Pr[ X \ge \lambda \cdot \Exp[X]] \le \exp\left(- p \cdot \Exp[X] \cdot c(\lambda)\right)$ for any $\lambda \ge 1$, and

    \item  $\Pr[ X \le \lambda \cdot \Exp[X]] \le \exp\left(- p \cdot \Exp[X] \cdot c(\lambda)\right)$ for any $0 < \lambda \le 1$.
  \end{enumerate}
\end{theorem}

Finally, we recall a simple version of a classical drift theorem.
\begin{theorem}[From Theorem 2.5 in~\cite{kotzing_theory_2024}] \label{thm:multiplicative_drift}
  Let~$(X_t)_{t \geq 0}$ be an integrable process taking values in~$\{0,\ldots,n\}$, and let~$T = \inf\{t \geq 0, X_t = 0\}$. Assume that there is~$\delta > 0$ such that for all~$t < T$,
  \begin{equation*}
    \Exp[X_t - X_{t+1} \mid X_t] \geq \delta \, X_t.
  \end{equation*}
  Then~$T \in O(\delta^{-1} \, \log n)$ with high probability.
\end{theorem}

Throughout, we implicitly use the following observation that tail bounds on the stabilisation time yield upper bounds on the expected stabilisation time for our protocols.
Let $S \subseteq \Lambda^n$ be a set of configurations.
We say that a protocol is $S$-closed if any execution $(x_t)_{t \ge 0}$ of the protocol on an input instance $(G,f)$ satisfies
\[
  x_0 \in S \implies x_t \in S \text{ for all } t \ge 0.
\]
That is, any execution started in $S$ stays in $S$.

\begin{lemma}\label{lemma:whp-to-exp}
  Let $F \colon \mathbb{N} \to \mathbb{N}$ be an increasing function and $k>1$ be a constant such that for any $x \in S$ the stabilisation time of the protocol satisfies
  \[
    \Pr[T(x) > F(n)] \le 1/n^{k}
  \]
  for all sufficiently large $n$.
  If the protocol is $S$-closed, then for any $x \in S$,
  \(
    \E[T(x)] \in O(F(n)).
  \)
\end{lemma}
\begin{proof}
  Fix $x \in S$, $F=F(n)$, and $p=1/n^k$. Because the protocol is $S$-closed, using the Markov property of the execution, we get for every $i \ge 0$ that
  \[
    \Pr[T(x) > (i+1)F \mid T(x) > iF]  \le p.
  \]
  Hence, by an easy induction, we get for each $i \ge 0$ that
  \[
    \Pr[T(x) > iF] \le p^i
  \]
  As the stabilisation time is a non-negative integer-valued random variable,
  we have
  \begin{align*}
    \E[T(x)]
    &= \sum_{t \ge 0} \Pr[T(x) > t] \\
    &\le \sum_{i \ge 0} \sum_{t=iF}^{(i+1)F-1} \Pr[T(x) > iF]
    \le F \sum_{i \ge 0} p^i \le F \cdot \frac{1}{1-p} \in O(F(n)). \qedhere
  \end{align*}
\end{proof}

\end{document}